\documentclass[
notitlepage,
floatfix,
aps,
prl,
reprint,
twocolumn,
superscriptaddress,
nofootinbib]{revtex4-2}
\usepackage[utf8]{inputenc}
\usepackage[normalem]{ulem}

\usepackage[caption=false]{subfig}
\usepackage{graphicx}
\usepackage{epstopdf}
\usepackage{array}
\usepackage{verbatim}
\usepackage{amsmath,amsfonts,amssymb,amscd,mathtools,amsthm}
\usepackage{tabularx}
\usepackage{stmaryrd}
\usepackage{enumerate}
\usepackage{wasysym}
\usepackage{mathrsfs}
\usepackage{microtype}
\usepackage{hyperref}
\usepackage[dvipsnames]{xcolor}
\usepackage{physics}
\usepackage{enumitem}
\hypersetup{
	bookmarksnumbered=true, 
	unicode=false, 
	pdfstartview={FitH}, 
	pdftitle={}, 
	pdfauthor={}, 
	pdfsubject={}, 
	pdfcreator={}, 
	pdfproducer={}, 
	pdfkeywords={}, 
	pdfnewwindow=true, 
	colorlinks=true, 
	linkcolor=NavyBlue, 
	citecolor=NavyBlue, 
	filecolor=NavyBlue, 
	urlcolor=NavyBlue 
}
\usepackage{tikz}
\usepackage[lmargin=.7in,rmargin=.7in,tmargin=.7in,bmargin=1in]{geometry}

\usepackage{newtxtext,newtxmath}
\usepackage{multirow}

\makeatletter
\def\maketitle{
	\@author@finish
	\title@column\titleblock@produce
	\suppressfloats[t]}
\makeatother

\newcommand{\suppl}{Supplemental Materials}

\usepackage{booktabs}

\theoremstyle{definition}
\newtheorem{thm}{Theorem}

\newtheorem*{central-fact}{Central Fact}

\theoremstyle{definition}

\newtheorem{lemma}[thm]{Lemma}

\newcommand{\eq}[1]{(\hyperref[eq:#1]{\ref*{eq:#1}})}

\renewcommand{\sec}[1]{\hyperref[sec:#1]{Section~\ref*{sec:#1}}}
\newcommand{\thrm}[1]{\hyperref[thrm:#1]{Theorem~\ref*{thrm:#1}}}
\newcommand{\lemm}[1]{\hyperref[lemm:#1]{Lemma~\ref*{lemm:#1}}}
\newcommand{\prop}[1]{\hyperref[prop:#1]{Proposition~\ref*{prop:#1}}}
\newcommand{\corr}[1]{\hyperref[corr:#1]{Corollary~\ref*{corr:#1}}}
\newcommand{\fig}[1]{\hyperref[fig:#1]{~\ref*{fig:#1}}}
\newcommand{\deff}[1]{\hyperref[deff:#1]{~\ref*{deff:#1}}}

\newcommand{\mB}{\mathcal{B}}
\newcommand{\mC}{\mathcal{C}}
\newcommand{\mD}{\mathcal{D}}

\newcommand{\mO}{\mathcal{O}}

\newcommand{\mS}{\mathcal{S}}

\newcommand{\mV}{\mathcal{V}}

\DeclareMathOperator{\id}{id}

\DeclareMathAlphabet{\mathmybb}{U}{bbold}{m}{n}
\newcommand{\1}{\mathmybb{1}}

\newcommand{\dm}[1]{\ketbra{#1}{#1}}

\newcommand{\bal}{\begin{align}\begin{aligned}}
\newcommand{\eal}{\end{aligned}\end{align}}

\newcommand{\lset}{\left\{ }
\newcommand{\rset}{\right\}}

\newcommand{\SSp}[1]{\mD_{#1}}
\newcommand{\CSp}[2]{\mathrm{CPTP}_{#1 \to #2}}
\newcommand{\FSSp}[1]{\mS_{#1}}
\newcommand{\FCSp}[2]{\mO_{#1 \to #2}}

\newcommand{\mincomp}[2]{\FSSp{#1}\!\otimes_{\text{min}}\!\FSSp{#2}}
\newcommand{\maxcomp}[2]{\FSSp{#1}\!\otimes_{\text{max}}\!\FSSp{#2}}
\newcommand{\sepcomp}[2]{\FSSp{#1}\!\otimes_{\text{sep}}\!\FSSp{#2}}
\newcommand{\affcomp}[2]{\FSSp{#1}\!\otimes_{\text{aff}}\!\FSSp{#2}}

\newcolumntype{L}[1]{>{\raggedright}p{#1}}
\newcolumntype{C}[1]{>{\centering}p{#1}}
\newcolumntype{R}[1]{>{\raggedleft}p{#1}}
\newcolumntype{D}{>{\centering\arraybackslash}X}

\usepackage{soul}
\usepackage{lipsum}
\usepackage{tcolorbox}
\tcbuselibrary{breakable}
\newtcolorbox{mybox}[1]{colback=red!5!white,colframe=red!65!black,fonttitle=\bfseries,title=#1,boxrule=0.7pt, breakable}
\usepackage{makecell}

\newtheorem{example}{Example}

\begin{document}
\title{Catalytic channels are the only noise-robust catalytic processes}
\author{Jeongrak Son}
\affiliation{School of Physical and Mathematical Sciences, Nanyang Technological
University, 21 Nanyang Link, 637371 Singapore, Republic of Singapore}

\author{Ray Ganardi}
\affiliation{School of Physical and Mathematical Sciences, Nanyang Technological
University, 21 Nanyang Link, 637371 Singapore, Republic of Singapore}

\author{Shintaro Minagawa}
\affiliation{Graduate School of Informatics, Nagoya University, Furo-cho, Chikusa-Ku, Nagoya 464-8601, Japan}
\affiliation{Institute of Systems and Information Engineering, University of Tsukuba, 1-1-1, Tennodai, Tsukuba, Ibaraki 305-8573, Japan}
\affiliation{Aix-Marseille University, CNRS, LIS, 13288 Marseille CEDEX 09, France}

\author{Francesco Buscemi}
\affiliation{Graduate School of Informatics, Nagoya University, Furo-cho, Chikusa-Ku, Nagoya 464-8601, Japan}

\author{Seok Hyung Lie}
\email{seokhyung@unist.ac.kr}
\affiliation{Department of Physics, Ulsan National Institute of Science and Technology (UNIST), Ulsan 44919, Republic of Korea}

\author{Nelly H.Y. Ng}
\email{nelly.ng@ntu.edu.sg}
\affiliation{School of Physical and Mathematical Sciences, Nanyang Technological
University, 21 Nanyang Link, 637371 Singapore, Republic of Singapore}
\affiliation{Centre for Quantum Technologies, Nanyang Technological University, 50 Nanyang Avenue, 639798 Singapore}
\date{\today}

\begin{abstract}
Catalysis refers to the possibility of enabling otherwise inaccessible quantum state transitions by supplying an auxiliary system, provided that the    auxiliary is returned to its initial state at the end of the protocol. 
We show that previous studies on catalysis are largely impractical, because even small errors in the system’s initial state can irreversibly degrade the catalyst. To overcome this limitation, we introduce \textit{robust catalytic transformations} and explore the fundamental extent of their capabilities. 
We demonstrate that robust catalysis is closely tied to the property of resource broadcasting.
In particular, in completely resource non-generating theories, robust catalysis is possible if and only if resource broadcasting is possible. 
We develop a no-go theorem under a set of general axioms, demonstrating that robust catalysis is unattainable for a broad class of quantum resource theories. However, surprisingly, we also identify thermodynamical scenarios where maximal robust catalytic advantage can be achieved. 
Our approach clarifies the practical prospects of catalytic advantage for a wide range of quantum resources, including entanglement, coherence, thermodynamics, magic, and imaginarity. 
\end{abstract}

\maketitle

\emph{Introduction}---%
The preparation of quantum states constitutes a central task of quantum information processing.
To this end, the ability to couple a quantum system to additional auxiliaries provides a way to activate fundamental performance advantages that are otherwise unattainable.
One obvious reason is that auxiliaries may inject extra quantum resources, such as energy, entanglement, or correlations.
A particularly intriguing scenario is when the advantages can be reaped even if the auxiliary system is used \emph{catalytically}, akin to those of chemical processes, where the final state of the auxiliary system is identical to its initial state~\cite{jonathan1999entanglement}.
This is interesting since it means that a one-off investment in the preparation of such catalysts provides a long-term, sustainable advantage.

Catalysis has proven to be surprisingly useful across various fields~\cite{LipkaBartosik2023CatReview,Datta2023CatReview}.
Classical catalytic computation~\cite{Buhrman2014CatalyticSpace, Buhrman2018CatSpace2, Cook2025CatalyticTree3} led to unexpected space-efficient algorithms and a space-time trade-off~\cite{Cook2020CatalyticTree, Cook2024CatalyticTree2, Williams2025SimulatingTime}, while quantum catalysis enables efficient quantum circuit compilations~\cite{Amy2023CatalyticEmbedding, Takeuchi2024MBQCCat, Kissinger2024CatalysingCompleteness, Buhrman2025QuantumCatalyticSpace, Khattar2025ConditionallyClean, Kim2025CatalyticzRotation, Gidney2025Adder}, including the seminal Cirac-Zoller gate~\cite{Cirac1994_Catalysis, Cirac1995CiracZollerGate}.
Catalysts also enhance quantum thermal machines~\cite{Ghosh2017CatalysisEngine, Henao2021catalytic, Henao2023Catalytic, Biswas2024Catalytic, Lobejko2024Catalytic}, without consuming additional resources.

Beyond their practical value, it is inherently interesting and puzzling that the mere presence of a catalyst can yield substantial advantages. 
To rigorously understand these gains, one must first establish the fundamental limits of what is achievable without catalysts.
Quantum resource theories~\cite{Chitambar2019QRTReview} provide a powerful, unified framework for this purpose, allowing systematic analysis of resources such as entanglement~\cite{Horodecki2009EntanglementReview}, coherence~\cite{Streltsov2017CoherenceReview}, and magic~\cite{Bravyi2005Magic}.
By defining a clear dichotomy between free and infeasible operations, resource theories facilitate precise comparisons between non-catalytic and catalytic capabilities for a given task. 
Since the initial example of enhancing entanglement manipulation~\cite{jonathan1999entanglement}, catalytic transformations have been shown to convert the majorization relations into trumping relations characterized by a set of inequalities~\cite{Turgut2007Trumping, Auburn2008Trumping, Klimesh2007Trumping, Brandao2015_2ndlaws}, relate one-shot to asymptotic transformations~\cite{Duan2005_DuanState, Wilming2021ReversibleCatalysis, Shiraishi2021GP, Kondra_2021}, spontaneously break symmetry~\cite{Stephen2025ManybodyCat}, and even enable arbitrary amplification of resources related to sensitivity in quantum metrology~\cite{Ding_2021,Takagi2022CorrCat, Shiraishi_2024,Kondra_2024,Zhang_2024}.

To this day, catalytic advantages have been demonstrated almost exclusively under idealized, noise-free conditions.
Practical catalytic advantage, however, critically depends on resilience to implementation errors.
Progress has been impeded by the fact that even the smallest deviations in the catalyst output state can compromise the catalytic paradigm.
This vulnerability is most evident in embezzlement~\cite{PhysRevA.67.060302,leung2014characteristics,ng2015limits}, where an arbitrarily small change in the catalyst state can enable pathological behavior, making any transformation catalytically possible. 

There are three primary sources of errors in catalytic protocols (see Fig.~\ref{fig:illustration}):
(i) errors in the implementation of the channel $\Lambda$, which can only be mitigated in a limited sense;
(ii) errors in preparing the system state $\rho$;
and (iii) errors in preparing the catalyst state $\tau$. 
The third case is easily managed, since catalyst preparation is a one-off occurrence: 
the data-processing inequality ensures that the final catalyst error remains bounded by the initial error.
However, when fresh errors occur in preparing system states $\rho_{\epsilon}$ during repeated use of a catalyst, those errors can accumulate, causing the catalyst to degrade.
In the worst case, the accumulated error on the catalyst can grow linearly with the number of repetitions. 
Thus, catalyst degradation due to preparation errors in the system is the more serious issue and will be the focus of this work.

A potential way to address system preparation errors would be to consider catalysis not only on the level of specific state transformations on the system, but \emph{catalytic implementations of channels}~\cite{Vidal_catchannel, Boes2018_randomness, Lie2021Generecity, Lie2021CatalyticRandomness}.
In fact, most examples in quantum computing~\cite{Buhrman2014CatalyticSpace, Buhrman2018CatSpace2, Cook2025CatalyticTree3, Cook2020CatalyticTree, Cook2024CatalyticTree2, Williams2025SimulatingTime, Amy2023CatalyticEmbedding, Takeuchi2024MBQCCat, Kissinger2024CatalysingCompleteness, Buhrman2025QuantumCatalyticSpace, Khattar2025ConditionallyClean, Kim2025CatalyticzRotation, Gidney2025Adder, Cirac1995CiracZollerGate} satisfy this notion as they focus on compiling (unitary) channels rather than merely state transformations.
Such channels can also be viewed as a special case of catalysis in dynamical resource theories~\cite{Devetak2008Dynamical, Chiribella2008Supermap, Rosset2018Dynamical, Liu2019Dynamical, Gour2019Dynamical, Takagi2020Dynamical, Gour2020Dynamical}, but systematic investigation in this direction remains largely unexplored~\cite{LipkaBartosik2023CatReview,Datta2023CatReview}. However, catalytic channels have insofar been perceived as an overly conservative way of characterizing catalytic advantage: the preservation of the catalyst regardless of input state means that one cannot make any strategic choices of fine-tuning the catalyst according to knowledge of the input. 

\begin{figure}
	\centering
	\includegraphics[width=\linewidth]{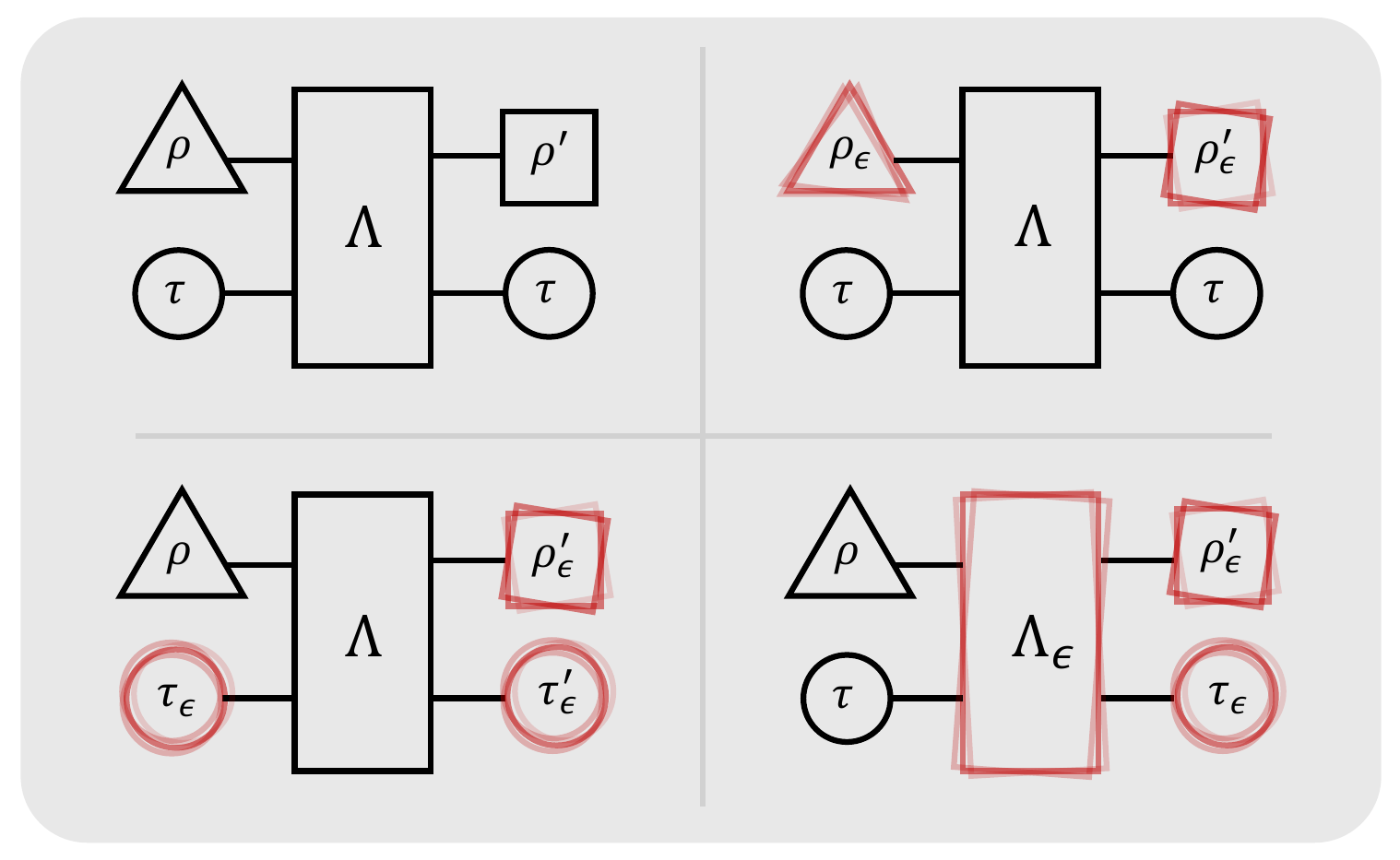}
	\caption{Illustrations of noisy catalysis.
		The top-left quadrant depicts the ideal scenario, where a catalyst $\tau$ facilitates a state transformation from $\rho$ into $\rho'$ and is returned exactly to its initial state.
		The top-right quadrant represents robust catalysis, where the catalyst state remains unaffected despite errors in the system state preparation. 
		The bottom quadrants show how other sources of errors (catalyst state preparation and channel implementation) impact the final states.}\label{fig:illustration}
\end{figure}

Our first result identifies catalytic channels as the only form of catalysis robust to state-preparation noise, even when the noise is assumed to be arbitrarily small. 
In other words, all other catalytic transformations are fragile in realistic settings. 
Therefore, it becomes important to reexamine catalytic advantages through the lens of catalytic channels. 
We proceed to identify general axioms that establish a fundamental no-go theorem for robust catalytic advantage; this result encompasses prominent resource theories such as athermality~\cite{Ng18_Qthermo_book, Lostaglio19_review} and coherence~\cite{Streltsov2017CoherenceReview}.
Nevertheless, we also show that robust catalysis is not entirely precluded.
We formulate a theory in the context of thermodynamics that maximally exhibits a robust catalytic advantage.
In this case, we demonstrate that catalytic state-preparation channels are fully characterized by a single resource measure, thereby enabling practical and significant catalytic benefits.

We have developed two novel techniques to derive these results. 
First, we establish the precise form of equivalence between robust catalysis and resource broadcasting~\cite{Lu2013FisherBroad, Marvian2019Broadcasting, Lostaglio2019Broadcasting, Yang_2021Broad, Zhang2024_magicnobroad}, the latter being a generalization of quantum state broadcasting~\cite{Barnum1996Broad,superbroad1,superbroad2,Barnum2007Broad, Piani2008Broad, Parzygnat2024VirtualBroad}. 
Consequently, all our results (no-go and constructive findings) also apply to resource broadcasting. 
Furthermore, we introduce a framework to analyze how different systems compose in resource theories and identify the extreme cases.
This compositional structure has been almost never explicitly explored in prior work, yet remarkably we find that the potential for catalytic channel advantage is governed by it rather than by the specifics of the resources themselves.
We anticipate both techniques will be of independent interest to the broader resource theory community. 

\medskip
\emph{Definitions and basic properties}---%
We denote quantum systems by $X,Y$, the set of density matrices on $X$ by $\SSp{X}$, and the set of quantum channels (trace-preserving completely positive linear maps) from $X$ to $Y$ by $\CSp{X}{Y}$. 
A \emph{catalytic transformation} is a process where a channel $\Lambda\in\CSp{SC}{S'C}$ is applied to the system and the catalyst $\rho_{S}\otimes\tau_{C}$ and preserves the catalyst state $\Tr_{S'}[\Lambda(\rho_{S}\otimes\tau_{C})] = \tau_{C}$. 
For generality, we consider correlated catalysis~\cite{LipkaBartosik2023CatReview, Datta2023CatReview}, but most of our results apply to strict catalyses where a stronger condition $\Lambda(\rho_{S}\otimes\tau_{C}) = \tilde{\rho}_{S'}\otimes\tau_{C}$ must hold. 
We define a catalytic transformation to be \emph{robust} (against $\epsilon$ initial-state-preparation noise) if for all system states $\sigma_{S}$ with $\Vert \sigma_{S} - \rho_{S}\Vert_{1}\le\epsilon$, the catalyst is recovered exactly, i.e. $\Tr_{S'}[\Lambda(\sigma_{S}\otimes\tau_{C})] = \tau_{C}$, for a parameter $\epsilon>0$.
Note that this definition makes no assumptions about the system's state after the process and the robustness parameter $\epsilon$, capturing the degree of initial state preparation errors, can be arbitrarily small.
Finally, we define \emph{catalytic channels} to be a channel $\tilde{\Lambda}(\bullet_S) \coloneq \Tr_{C}[ \Lambda (\bullet_{S}\otimes\tau_{C})]$ such that the catalyst is preserved for any input state in $S$, i.e. $\Tr_{S'}[\Lambda(\sigma_{S}\otimes\tau_{C})] = \tau_{C}$ for any $\sigma_{S}\in\SSp{S}$.
They are typically studied when the dilation $\Lambda$ is a unitary operation~\cite{Boes2018_randomness, Lie2021Generecity, Lie2021CatalyticRandomness}, but here we allow $\Lambda$ to be any channel.
Note that catalytic channels are by definition robust for any parameter $\epsilon>0$%
\footnote{We stress that catalytic channels should not be confused with \emph{universal catalysis}, where universality is achieved by fine-tuning the free operations with respect to input states~\cite{LipkaBarosik2021_universal, Datta2024_universal}.}.

As previously mentioned, we define the robustness as the exact recovery of the catalyst to prevent embezzlement and the continual degradation of the catalyst after repeated usages. 
In experiments, it would be hard to verify whether the catalyst is truly identical to its initial state.
However, our definition provides a theoretical framework for catalytic advantages that would be obtained with advancements in experimental techniques (i.e. robust catalysis) as opposed to the ones that are unattainable due to their inherent fragility. 
In other words, robustness must be the guiding principle of a practical catalytic advantage in experiment. 

\medskip
\emph{Only catalytic channels are robust catalysis}---%
Let us start with a central fact that even the minimalistic requirements set for robust catalysis single out catalytic channels.
\begin{central-fact}
	For any initial state $\rho_S$ and any $\epsilon>0$, a catalytic transformation with a channel $\Lambda\in\CSp{SC}{S'C}$ and the catalyst $\tau_{C}$ is robust against $\epsilon$ initial-state-preparation noise if and only if the resulting channel $\tilde{\Lambda}(\bullet_S) \coloneq \Tr_{C}[ \Lambda (\bullet_{S}\otimes\tau_{C})]$ is a catalytic channel.  
\end{central-fact}

The proof idea is that the robustness condition restricts the channel action on a full-dimensional subset of inputs (full argument in End Matters).
This Central Fact clarifies that all forms of catalysis, other than catalytic channels, are inevitably fine-tuned to a very specific initial state of the system and risk degrading the catalyst whenever the system's state is not prepared with strictly infinite precision.
It also ties the robustness of catalysis in resource theories of \emph{states} to catalysis in resource theories of \emph{channels}, in which catalytic channels are included.

The Central Fact prompts us to examine whether catalytic channels provide any meaningful advantage. 
Besides the clever yet serendipitous findings in circuit compilation~\cite{Amy2023CatalyticEmbedding, Takeuchi2024MBQCCat, Kissinger2024CatalysingCompleteness, Buhrman2025QuantumCatalyticSpace, Khattar2025ConditionallyClean, Kim2025CatalyticzRotation, Gidney2025Adder, Cirac1995CiracZollerGate}, two universal mechanisms for catalytic channel advantages are known.
First, populations of the catalyst state can always be accessed catalytically and used as a randomness source for implementing convex combinations of operations~\cite{Boes2018_randomness, Lie2021CatalyticRandomness, Lie2021Generecity}. 
Moreover, when the channel that resets the catalyst state is available, any catalytic transformation can be made robust by applying the reset channel at the end; such channels are found to be useful in thermodynamic resource theories~\cite{Korzekwa2022_optimizing, Czartowski2023_thermalrecall, Son2024_CETO, Son2024hierarchy}.
Nevertheless, catalytic channels extend beyond these two scenarios, and the true extent of the robust catalytic advantage warrants further scrutiny.

\medskip
\emph{Technical tools: resource theories, compositions, and resource broadcasting}---%
For the systematic study of robust catalytic advantages, we work in the framework of resource theories. 
A resource theory is characterized by two elements: a set of free states $\mS$ and free operations $\mO$. 
System subscripts specify the system, e.g. $\FSSp{X}$ for free states on $X$ and $\FCSp{X}{Y}$ for free operations. 
We assume free states are a fixed property of each system; for instance, in the athermality theory, qubits with different Hamiltonians are in distinct systems.
To focus on advantages beyond those already studied, we consider free operations that are already convex and thus not trivially enhanced by catalysts.
Hence, we make the following basic assumptions about the set of free states and operations:
\begin{enumerate}[topsep=3pt,itemsep=0ex,leftmargin=*,label=(A\arabic*)]
	\item $\rho_{A}\otimes\rho_{B}\in\FSSp{AB}$ whenever $\rho_{A}\in\FSSp{A}$ and $\rho_{B}\in\FSSp{B}$, 
	\item if $\rho_{AB}\in\FSSp{AB}$, then $\Tr_{B}[\rho_{AB}]\in\FSSp{A}$ and $\Tr_{A}[\rho_{AB}]\in\FSSp{B}$,
	\item $\FSSp{X}$ is a convex set for any system $X$, and
	\item there always exists a full-rank state $\gamma\in\FSSp{}$.
\end{enumerate}
The first three are often used as basic assumptions, while (A4) is justified in \suppl~Sec.~\ref{subsection:non-fullrank}. 
Nevertheless, they cover a wide range of quantum resource theories, including those of entanglement, athermality, coherence, asymmetry, and magic. 
Unless otherwise specified, we consider \emph{completely resource non-generating (CRNG) theories}, where the free operations include all channels $\Lambda\in\CSp{A}{B}$, such that 
\begin{align}\label{eq:CRNG}
\left(\id_{A'}\otimes\Lambda\right)(\gamma_{A'A}) \in\FSSp{A'B},
\end{align}
for any system $A'$ and all $\gamma_{A'A}\in\FSSp{A'A}$.
Here, $\id_{A'}\in\CSp{A'}{A'}$ denotes the identity channel. 
CRNG operations are the full set of channels that cannot generate any resource from free states, even when acting on subsystems. 
Allowing free operations beyond CRNG often trivializes the theory by enabling arbitrary state transformations.
Well-known examples of CRNG operations include separable operations for entanglement~\cite{Vedral1997SEP, Rains1998SEP}, Gibbs-preserving operations for athermality~\cite{Faist2015GP}, and covariant operations for asymmetry~\cite{Keyl1999Cov, Gour2008Cov}.

For catalytic channels to provide advantage in CRNG resource theories, a dilated channel $\Lambda\in\CSp{SC}{S'C}$ that acts on the system and the catalyst must be free, while the resulting catalytic channel $\tilde{\Lambda}\in\CSp{S}{S'}$ must not. 
Therefore, when considering CRNG free operations, it suffices to find a free state $\gamma_{S}$ such that $\tilde{\Lambda}(\gamma_{S})\notin\FSSp{S'}$. 
Conversely, if all catalytic channels $\tilde{\Lambda}$ map $\FSSp{S}$ to $\FSSp{S'}$, then no advantage can be gained robustly\footnote{Strictly speaking, this would prove that any catalytic channel is always resource non-generating (RNG). 
However, if all catalytic channels $\tilde{\Lambda}$ are RNG, then their extensions $\id_{A}\otimes\tilde{\Lambda}$ are also catalytic channels, and therefore RNG, which in turn implies that all $\tilde{\Lambda}$ are in fact \textit{completely} RNG as well.}.
This property has also been studied in terms of \emph{resource broadcasting}~\cite{Lu2013FisherBroad, Marvian2019Broadcasting, Lostaglio2019Broadcasting, Yang_2021Broad, Zhang2024_magicnobroad}, where a free operation $\mB\in\FCSp{A}{AB}$ may ``broadcast'' some resource from $A$ to another system $B$, i.e. $\Tr_{A}[\mB(\rho_{A})]\notin\FSSp{B}$, while leaving $A$ fully intact, i.e., such that $\Tr_{B}[\mB(\rho_{A})] = \rho_{A}$.
The outcome of a catalytic channel applied to a free system state $\gamma_{S}\in\FSSp{S}$, is therefore closely related to resource broadcasting, as one wants $\tilde{\Lambda}(\gamma_{S}) = \Tr_{C}[\Lambda(\gamma_{S}\otimes\tau_{C})]\notin\FSSp{S'}$. 

\begin{thm}\label{thm:broadcast_robcat}
	For any CRNG resource theory that satisfies Assumptions (A1), (A2), and (A3), the existence of a catalytic channel $\tilde{\Lambda}\notin\FCSp{S}{S'}$ is equivalent to the existence of resource broadcasting.
\end{thm}

Thm.~\ref{thm:broadcast_robcat} is proven in the End Matters; here we only comment about a subtle difference between catalytic channels and resource broadcasting. 
In real experimental setups, implementing the broadcasting channel $\mB\in\FCSp{C}{S'C}$ requires an auxiliary state $\gamma_{S'}\in\FSSp{S'}$, which is assumed to be fixed. 
Unlike catalytic channels, the broadcasting process may fail to preserve the original state $\tau_{C}$ when the auxiliary state $\gamma_{S'}$ is somehow perturbed. 

We now address the central question: \emph{exactly when do catalytic channels offer a net advantage?}
It turns out that the answer hinges on an additional degree of freedom that has not been explicitly considered before: how the composite free state set $\mathcal{S}_{AB}$ is related to $\mathcal{S}_{A}$ and $\mathcal{S}_{B}$.
In typical resource theories, this composition is usually defined operationally, depending on the particular resource at hand.
However, theoretically, we can freely choose any composition rule as long as it is consistent with the rest of the theory.
Remarkably, we can show that any choice $\FSSp{AB}$ satisfying Assumptions (A1)-(A3) must lie between \textit{minimal} and \textit{maximal} compositions, denoted $\mincomp{A}{B}$ and $\maxcomp{A}{B}$, which, for given $\FSSp{A}$ and $\FSSp{B}$, are explicitly written as
\begin{align}
	\!\!\!\mincomp{A}{B} &\coloneq \mathrm{conv}\lset \rho_{A}\otimes\rho_{B} \,\vert\, \rho_{A}\in\FSSp{A},\ \rho_{B}\in\FSSp{B} \rset,\label{def:min_comp}\\
	\!\!\!\maxcomp{A}{B} &\coloneq \lset \rho_{AB} \,\vert\, \Tr_{B}[\rho_{AB}]\in\FSSp{A},\Tr_{A}[\rho_{AB}]\in\FSSp{B}\rset.\label{def:max_comp}
\end{align}
Eqs.~\eqref{def:min_comp} and~\eqref{def:max_comp} evoke the tensor product of convex cones (explained in End Matters) and encapsulate the limits of allowed correlations in free states.
For instance, the resource theories of athermality~\cite{Ng18_Qthermo_book, Lostaglio19_review}, where the tensor product of subsystem Gibbs states is a free state, and coherence~\cite{Streltsov2017CoherenceReview}, where diagonal states are free, follow the minimal composition rule.
On the other hand, the maximal composition of athermality includes the thermofield double state (also known as the two-mode squeezed vacuum state)~\cite{Israel1976Thermofield}---a pure entangled state with Gibbs state marginals.
While $\maxcomp{A}{B}$ may appear contrived, it can be understood as a theory concerned with \emph{local}, rather than global, resources.
These extremal compositions streamline the analysis by making the set of CRNG operations identical to RNG operations, eliminating the need to consider resource-generating effects on larger Hilbert spaces as in Eq.~\eqref{eq:CRNG} (see Lemma~\ref{lem:RNGCRNG} in \suppl).
We note that Ref.~\cite{Pinske2024Censorship} defines $\mincomp{A}{B}$ as their composition rule in the context of resource censorship.

\medskip
\emph{When robust catalytic advantage is impossible}---With this newfound categorization of compositions, we establish a no-go theorem for robust catalytic advantages when the composition restricts correlations between the partitions of free states, i.e. when the free state set is minimal. 
\newcommand{\mbM}{\mathbb{M}}
\newcommand{\REM}{R}  
\newcommand{\MREM}{R_{\mbM}}
\begin{thm}\label{thm:mincomp_nocat}
	If a convex CRNG resource theory has minimal composition, it allows neither resource broadcasting nor any robust catalytic advantage.
\end{thm}
The proof boils down to deriving an inequality resembling strong super-additivity, where a composite state's resource exceeds the sum of its marginal resources.
This rules out resource broadcasting; see \suppl~Sec.~\ref{app:proof_thm_mincomp_nocat} for full proof. 
While no-broadcasting was previously established specifically for theories of asymmetry under connected Lie groups~\cite{Marvian2019Broadcasting, Lostaglio2019Broadcasting} and stabilizer operations~\cite{Zhang2024_magicnobroad}, our result is the first known sufficient condition guaranteeing no-broadcasting across generic classes of resource theories.
Furthermore, we demonstrate many significant theories beyond the minimal composition class also lack robust catalytic advantage and disallow resource broadcasting (see End Matters for a list and \suppl~for proofs). 

\medskip
\emph{When robust catalytic advantages are possible}---%
Surprisingly, taking an alternative composition rule may enable useful catalytic channels and resource broadcasting.
Consider the theory of local athermality, where we identify all states that are locally thermal as free states $\mathcal{S}_{SC} = \Bqty{\mu_{SC} \,|\, \Tr_B \mu_{SC} = \gamma_C, \Tr_C \mu_{SC} = \gamma_B }$, where $\gamma_X$ denotes the thermal state of system $X$.
Note that this is the maximal composition $\FSSp{SC} = \maxcomp{S}{C}$ of $\FSSp{S} = \Bqty{\gamma_S}$ and $\FSSp{C} = \Bqty{\gamma_C}$, and it may include entangled states, as long as all the marginals of the state are locally thermal.
We can define an analogue of free energy through the max-relative entropy~\cite{Datta2009_maxrel} $F_{\max}(\rho_S) \coloneq D_{\max}(\rho_S \| \gamma_S)$.
Below, we show that $F_{\textrm{max}}$ characterizes the catalytic state preparation channels that bring robust catalytic advantages.

\begin{thm}\label{thm:robcat_possible}
	In the local athermality theory, catalytic measure-and-prepare channels $\tilde{\Lambda}(\bullet_{S}) = \sigma_{S}\Tr[\bullet_{S}]$ can be implemented with some catalyst $\tau_{C}$ if and only if 
  \begin{align}\label{eq:dmax_cond}
    \sup_{\tau_{C}\in\SSp{C}} F_{\max} (\tau_C)
    \geq
    F_{\max} (\sigma_S).
  \end{align}
\end{thm}
We use the properties of $F_{\max}$ to prove the only if direction, while for the if direction we explicitly construct the catalytic channel $\tilde{\Lambda}$.
The full proof of Thm.~\ref{thm:robcat_possible} along with its generalizations and limitations are provided in \suppl~Section~\ref{section: robcat possible generalization}.
There, we also show that the presence of free entangled states in the maximal composition does not matter: Theorem~\ref{thm:robcat_possible} still holds if one redefines the composition to exclude all entangled states from the maximal composition.
Here, we offer a few remarks. 
First, the dilation $\Lambda\in\FCSp{SC}{SC}$ of the catalytic channel $\tilde{\Lambda}$ can be used to construct a resource broadcasting channel $\mB(\bullet_{C}) = \Lambda(\gamma_{S}\otimes\bullet_{C})\in\FCSp{C}{SC}$ with the free state $\gamma_{S}\in\FSSp{S}$.
Hence, Thm.~\ref{thm:robcat_possible} also fully characterises resource broadcasting in the local athermality theory. 
Additionally, when the smallest eigenvalue of $\gamma_{S}$ is greater than that of $\gamma_{C}$, transformations between any states become feasible as Eq.~\eqref{eq:dmax_cond} becomes trivially true, effectively trivializing the theory.
Note that catalytic channels (that are not measure-and-prepare channels) from $\rho_{S}$ to $\sigma_{S}$ may still exist even when $\sigma_{S}$ violates Eq.~\eqref{eq:dmax_cond}.
Furthermore, we note that powerful examples of catalytic channels in the resource theory of imaginarity are found in Refs.~\cite{Takagi2017Imaginarity, Zhang_2024a}; see \suppl~for discussion.

\medskip
\emph{Discussion}---%
We show that catalytic channels are the only way to achieve catalytic advantage in the presence of unstructured noise. This result decisively establishes catalytic channels as the primary focus for future developments of catalysis, especially in the context of practical, noise-affected implementations. We rigorously identify fundamental limitations governing the (im)possibility of such advantages under CRNG operations. In particular, our no-go theorem shows that practical catalytic advantage cannot be expected for many CRNG theories. 

Our results suggest that the most promising, robust, and practically significant instances of catalysis should emerge in non-CRNG frameworks.
Finding such catalytic examples in commonly studied frameworks, e.g. LOCC or stabilizer operations, would be highly interesting and ripe for experimental demonstrations. 
Nevertheless, any catalytic advantage within such operationally defined theories is necessarily capped by the limitations of their corresponding CRNGs.
Moreover, since no-broadcasting in CRNG implies no-broadcasting for its subsets, our findings establish a general upper bound on resource broadcastability in any well-defined resource theory. 
Conversely, observing resource broadcasting within non-CRNG free operations would indicate robust catalytic advantage.
Another interesting question concerns other types of advantages that catalytic channels can bring.
Beyond resource advantages such as the ability to implement a non-free operation, catalytic channels might provide dimensional advantages~\cite{Boes2018_randomness}, as quantum catalysts could outperform classical randomness.

Robust catalysis also paves an alternative pathway for investigating \emph{channel catalysis}, where a free channel induces a non-free one through the assistance of another channel that functions as a catalyst.
Our work reveals that catalytic channels are the result of specific channel catalyses, where the channel acting as a catalyst is a fixed-outcome preparation channel.
Generic channel catalysis, on the other hand, is not immediately precluded by our no-go theorems, as demonstrated in Ref.~\cite{Vidal_catchannel, Lipka-Bartosik2021Teleportation} for the entanglement theory.
These examples do not fit our definition of robust catalysis as the system-catalyst channel is not free; however, they also represent robust catalytic advantage as the final channel is attainable only when a catalyst is present while this advantage is robust under errors in initial state preparation.
Hence, channel catalysis hints at a new potential for robust catalytic advantages in important theories such as entanglement, coherence, and magic, despite our result ruling out the existence of catalytic channels in those theories.  
We leave this question for future investigations.

Lastly, we identify a hierarchy of composition rules for free states within resource theories, revealing a spectrum from the impossibility to the possibility of robust catalysis. 
This hierarchy effectively delineates the types of correlations permitted under free operations. 
We anticipate that this framework will foster new strategies for extending or even hybridizing resource theories, offering novel approaches to exploring complex resource interactions.

\medskip
\textit{Acknowledgments}---We thank Hayato Arai, Martti Karvonen, Tulja Varun Kondra, Kohdai Kuroiwa, Ryuji Takagi, and Henrik Wilming for helpful discussions.
J.S., R.G. and N.N. are supported through the start-up grant of the Nanyang Assistant Professorship of
Nanyang Technological University, Singapore. S.H.L. acknowledges the start-up grant of Ulsan National Institute of Science and Technology, South Korea. 
S.M. is supported by JASSO Scholarship for Study Abroad under Agreement, ``THERS Make New Standards Program for the Next Generation Researchers,'' JST SPRING, Grant Number JPMJSP2125, and JST ASPIRE, Grant Number JPMJAP2339, the French government under the France 2030 investment plan, as part of the Initiative d'Excellence d'Aix-Marseille Université-A*MIDEX, AMX-22-CEI-01.
F.B. acknowledges support from MEXT Quantum Leap Flagship Program (MEXT QLEAP) Grant No. JPMXS0120319794, from MEXT-JSPS  Grant-in-Aid for Transformative Research Areas  (A) ``Extreme Universe,'' No.~21H05183, and  from JSPS  KAKENHI Grant No.~23K03230.

\newpage

\normalsize
\renewcommand{\theequation}{A\arabic{equation}}
\setcounter{equation}{0}  

\onecolumngrid
\section*{End Matters}\label{sec:end_matter}
\emph{Appendix A: Table summarizing resource theories with/without robust catalysis.}
\begin{center}
	\begin{tabular}{|c|c|c|}
		\hline
		\multirow{2}{*}{} & \multicolumn{2}{c|}{\textbf{Robust Catalysis and Resource Broadcasting}} \\
		\cline{2-3} 
        &  &  \\
		& \textbf{Yes} & \textbf{No} \\
        &  &  \\
		\multirow{6}{*}{\textbf{\makecell{CRNG\\resource\\ theories}}} & Athermality ($T=0$)~\cite{Kuroiwa2020catreplication} & {Athermality ($T>0$)} [Thm.~\ref{thm:mincomp_nocat}] \cite{Wilming2017FreeE} \\
		& Imaginarity~\cite{Takagi2017Imaginarity,Zhang_2024a,Wu2021Imaginarity}  & {MIO Coherence} [Thm.~\ref{thm:mincomp_nocat}] \cite{Xi2015CoherenceSuperAdd}\\
		& Asymmetry (finite groups)~\cite{Marvian2013GCov} &  Entanglement~\cite{Christandl2004Squashed, Piani2009MRelEnt} \\
        & Thm.~\ref{thm:robcat_possible} &PPT entanglement~\cite{arxiv_Ganardi_2023a}\\
		&\multirow{2}{*}{\makecell{Limited subspace theories\\ $[$\suppl~Sec.~\ref{subsection:non-fullrank}$]$}} & Magic~\cite{Zhang2024_magicnobroad} \\
		&   & Asymmetry (connected Lie groups)~\cite{Marvian2019Broadcasting, Lostaglio2019Broadcasting} \\
		&   & Optical nonclassicality~\cite{Ferrari2023CVResources} \\
        &  &  \\
        \hline \hline
        \multirow{2}{*}{} & \multicolumn{2}{c|}{\textbf{Robust Catalysis}} \\
        \cline{2-3}
        &  &  \\
		& \textbf{Yes} & \textbf{No} \\
        & &  \\
        \multirow{3}{*}{\textbf{\makecell{non-CRNG\\resource\\ theories}}} & Elementary thermal operations~\cite{Son2024_CETO, Son2024hierarchy}  &  \multirow{2}{*}{\makecell{Gibbs preserving covariant operations ($T>0$)\\ $[$\suppl~Sec.~\ref{subsec:GPC}$]$}} \\
        &Markovian thermal operations~\cite{Korzekwa2022_optimizing,Czartowski2023_thermalrecall,Son2024hierarchy}&\\
        &Unitary operations~\cite{Boes2018_randomness, Lie2021CatalyticRandomness, Lie2021Generecity}& Thermal operations ($T>0$)~\cite{Lie2025Thermal} \\
        &  &  \\
        \hline
	\end{tabular}
\end{center}

\vspace{0.3cm}

\twocolumngrid
Entries in the no-broadcasting column have been shown to either prohibit broadcasting directly or have a strongly super-additive monotone in the corresponding references.
The references in the other column contain examples of robust catalysis. 

Generally, it is harder to show the non-existence of robust catalysis for non-CRNG operations. 
Hence, it remains unknown whether notable theories such as entanglement under LOCC or magic under stabilizer operations exhibit robust catalysis.
However, from their CRNG counterparts, the possibility of resource broadcasting is already ruled out (see Fig.~\ref{fig:RNGCRNGRCRB}).
It would also be interesting to extend no-broadcasting for theories that are not state-based~\cite{Karvonen2021Nonlocality}.
\begin{figure}[h!]
    \centering    \includegraphics[width=0.9\linewidth]{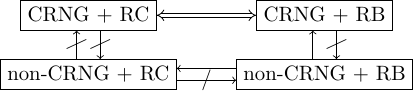}
    \caption{The logical relationship between the existence of robust catalysis (RC) and resource broadcasting (RB). The two-way relationship between RC and RB for completely resource non-generatic (CRNG) theories is shown by Thm.~\ref{thm:broadcast_robcat}. For non-CRNG theories, broadcasting implies robust catalysis, but not the other way around.}
    \label{fig:RNGCRNGRCRB}
\end{figure}

\vspace{0.5cm}

\emph{Appendix B: Strict catalysis and resource broadcasting}---%
So far, we have focused on catalysis wherein the catalyst may retain correlations with the system post-operation. 
Alternatively, we can define strict catalysis, where the catalyst must be recovered uncorrelated, i.e. $\Lambda(\rho_{S}\otimes\tau_{C}) = \rho'_{S'}\otimes\tau_{C}$. 
Strict catalysis has an advantage in the reusability of the catalysts, as they are recovered independently from the system, but its catalytic advantage is usually much more limited compared to the more general correlated catalysis~\cite{LipkaBartosik2023CatReview, Datta2023CatReview}.  
Strict catalysis is defined to be robust against $\epsilon$ initial-state-preparation noise if $\Lambda(\sigma_{S}\otimes\tau_{C}) = \sigma'_{S'}\otimes\tau_{C}$ for all $\sigma_{S}\in\SSp{S}$ such that $\Vert \sigma_{S} - \rho_{S}\Vert_{1}\le\epsilon$, where $\sigma'_{S'}$ is some state that depends on $\sigma_{S}$.
Similarly a strict catalytic channel is defined as a channel $\tilde{\Lambda}(\bullet_S) \coloneq \Tr_{C}[ \Lambda (\bullet_{S}\otimes\tau_{C})]$ such that $\Lambda(\sigma_{S}\otimes\tau_{C}) = \tilde{\Lambda}(\sigma_{S}) \otimes \tau_{C}$ for any $\sigma_{S}\in\SSp{S}$, and a strict resource broadcasting is defined as a process where a free operation $\mB\in\FCSp{A}{AB}$ outputs $\mB(\rho_{A}) = \rho_{A}\otimes\rho'_{B}$, where $\rho'_{B}\notin\FSSp{B}$.

All our theorems apply to the case of strict catalysis and resource broadcasting. 
For the Central Fact and~\ref{thm:broadcast_robcat}, we include the strict catalysis and resource broadcasting cases in the respective proofs. 
Thm.~\ref{thm:mincomp_nocat} already forbids strict robust catalysis and strict resource broadcasting, as these are specific instances of their correlated counterparts.
Moreover, we show that it is impossible to generate non-free strict catalytic channels in any resource theory if the catalyst state is full-rank, even when there exist non-free catalytic channels; see \suppl~for the proof.
Finally, Thm.~\ref{thm:robcat_possible} holds for strict catalytic channels and resource broadcastings, as our construction already recovers the catalyst without correlation.  

\vspace{0.3cm}

\emph{Appendix C: Proof of the Central Fact}---%
If a channel $\tilde{\Lambda}(\bullet_S) \coloneq \Tr_{C}[ \Lambda (\bullet_{S}\otimes\tau_{C})]$ is a catalytic channel, it implements, in particular, a robust catalysis, for any state $\rho_S$ and any robustness parameter $\epsilon>0$. We thus need to prove only the converse, and we do it by contradiction.
	Suppose that $\Lambda$ implements a $(\rho,\epsilon)$-robust catalysis, but is not a catalytic channel. This means that there exists at least one state, say $\eta_{S}$, such that
 \begin{align}
     \tau_C=\Tr_{S'}[ \Lambda (\rho_{S}\otimes\tau_{C})] \neq \Tr_{S'}[ \Lambda (\eta_{S}\otimes\tau_{C})]=:\tau_C'.
 \end{align}
	Let us then consider the state $\tilde\rho_{S} := (1-\frac{\epsilon}2)\rho_{S} + \frac{\epsilon}2\eta_{S}$. By construction,  we have $\Vert \tilde\rho_{S} - \rho_{S}\Vert_{1}=\Vert \frac{\epsilon}2\eta_{S}- \frac{\epsilon}2\rho_{S}\Vert_{1}\le\epsilon$. 
	Nevertheless, by linearity, $\Tr_{S'}[ \Lambda (\tilde\rho_{S}\otimes\tau_{C})] = (1-\frac{\epsilon}{2})\tau_{C} + \frac{\epsilon}{2}\tau'_{C} \neq\tau_{C}$, contradicting the assumption that $\Lambda$ implements an $\epsilon$-robust catalysis for $\rho_S$.
	The same argument work for strict version of robust catalysis, by analysing $\Lambda (\rho_{S}\otimes\tau_{C})$ instead of $\Tr_{S'}[ \Lambda (\rho_{S}\otimes\tau_{C})]$.
\vspace{0.3cm}

\emph{Appendix D: Proof of Theorem~\ref{thm:broadcast_robcat}}---%
Suppose that $\tilde{\Lambda}\notin\FCSp{S}{S'}$ is a catalytic channel with a dilation $\Lambda\in\FCSp{SC}{S'C}$ and a catalyst $\tau_{C}$.
Then there exists a state $\gamma_{S}\in\FSSp{S}$ such that the output $\rho_{SC} = \Lambda(\gamma_{S}\otimes\tau_{C})$ satisfies:
\begin{enumerate}
    \item the catalytic condition, $\Tr_{S}[\rho_{SC}] = \tau_{C}$, and
    \item the output on the system is a non-free state, $\Tr_{C}[\rho_{SC}]\notin\FSSp{S'}$.
\end{enumerate}
Now, define the channel $\mB\in\FCSp{C}{S'C}$ by $\mB(\bullet_{C}) = \Lambda(\gamma_{S}\otimes\bullet_{C})$.
This construction makes $\mB$ a broadcasting channel that maps $\tau_{C}$ into $\rho_{SC}$.
	
For the converse direction, suppose that there exists a broadcasting map $\mB\in\FCSp{C}{S'C}$ that broadcasts $\tau_{C}$. 
Then define the channel $\Lambda = \mB\circ\Tr_{S}\in\FCSp{SC}{S'C}$ which is free by our second basic assumption that partial traces are also free operations.
Consequently, $\tilde{\Lambda}(\rho_{S}) = \Tr_{C}[\Lambda(\rho_{S}\otimes\tau_{C})]$ becomes a catalytic channel such that $\tilde{\Lambda}(\gamma_{S}) = \Tr_{C}[\mB(\tau_{C})]\notin\FSSp{S'}$.

The same proof for both directions can be used for strict catalytic channels and resource broadcasting 
\vspace{0.3cm}

\emph{Appendix E: Minimal/maximal composition and tensor product of convex cones}---%
We begin by reviewing the minimal/maximal tensor products of convex cones~\cite{Aubrun_2021,deBruyn_2022,aubrun2022entanglement}, often discussed in the context of general probabilistic theories~\cite{janotta2014generalized,plavala2023general}.
First, we define convex cones and their duals following Refs.~\cite{boyd2004convex,aliprantis2007cones}.
Let $\mV$ be a vector space.
A non-empty subset $\mC\subset\mV$ is a convex cone if it is convex and closed under positive scalar multiplication.
We also assume that $\mC$ is closed and $\mC\cap(-\mC)=\{0\}$.
Let $\mV^*$ be the set of linear functionals on $\mV$ with the duality $\langle,\rangle:\mV^*\times\mV\to\mathbb R$.
The dual cone of $\mC$ is then given by $\mC^{*} = \Bqty{x^* \in\mV^*\,|\, \langle x^*, z \rangle \geq 0 \textrm{ for all } z \in \mathcal{C}}$.
The minimal and maximal tensor products of two cones $\mathcal{C}_{A}$ and $\mathcal{C}_{B}$ are defined as
\begin{align}
	\!\!\!\mC_{A}\!\otimes_{\textrm{min}}\!\mC_{B}
	&\coloneq
	\textrm{conv} \Bqty{
		z_A \otimes z_B
		\,|\,
		z_A \in \mC_{A},
		z_B \in \mC_{B}
	},\label{def:min_tensor_product}
	\\
	\!\!\!\mC_{A}\!\otimes_{\textrm{max}}\!\mC_{B}
	&\coloneq
	\Bqty{
		z
		\,|\,
		\langle x^*\otimes y^*, z \rangle \geq 0,\,
		x^* \in \mC^{*}_{A},\,
		y^* \in \mC^{*}_{B}
	}.  \label{def:max_tensor_product}
\end{align}

Consider the case where cones are defined on the space of linear Hermitian operators acting on Hilbert spaces, with the duality $\langle,\rangle$ given by the Hilbert-Schmidt inner product.
We now observe that for each free state set $\FSSp{A}$, the associated cone can be defined as $\mC(\FSSp{A}) = \bigcup_{\lambda\geq0}\lambda\FSSp{A}$.
Then, it follows that $\mC(\FSSp{A})\!\otimes_{\textrm{min}}\!\mC(\FSSp{B}) = \mC(\mincomp{A}{B})$, i.e. the minimal tensor product is equivalent to the minimal composition in Eq.~\eqref{def:min_comp}.

For the maximal counterparts, we have $\mC(\FSSp{A})\!\otimes_{\textrm{max}}\!\mC(\FSSp{B}) \subset \mC(\maxcomp{A}{B})$.
To see this, note that choosing $x^{*} = \Tr_{A}\in\mC^{*}(\FSSp{A})$ in Eq.~\eqref{def:max_tensor_product}, ensures the $B$ marginal of $\gamma_{AB}\in\mC(\FSSp{A})\!\otimes_{\textrm{max}}\!\mC(\FSSp{B})$ is always in $\mC(\FSSp{B})$.
Choosing $y^{*} = \Tr_{B}\in\mC^{*}(\FSSp{B})$ implies the same for the $A$ marginal.
The converse of this inclusion does not hold in general, as one can check by constructing counter-examples. 

Equipped with minimal and maximal composition rules for free states, it is tempting to construct minimal and maximal free operations for given subsystem theories, using CRNG operations. 
However, there is no natural inclusion between $\mathrm{CRNG}(\mincomp{A}{B})$ and $\mathrm{CRNG}(\maxcomp{A}{B})$, unlike the neat duality between minimal/maximal tensor products and their dual cones. 
We leave the problem of defining minimal and maximal composite free operations for future studies.

\clearpage
\newpage

\setcounter{page}{1}
\renewcommand{\thepage}{Supplemental Materials -- \arabic{page}/8}

\title{Supplemental Materials for ``Catalytic channels are the only noise-robust catalytic processes''}
\renewcommand{\theequation}{S\arabic{equation}}
\setcounter{equation}{0} 

\maketitle
\onecolumngrid

\section{Sources of catalytic errors}

Any catalytic protocol has three ingredients: the system state $\rho$, the channel $\Lambda$, and the catalyst $\tau$ (see Fig.~\ref{fig:illustration}). 
We briefly discuss possible sources of errors in a catalytic protocol, which may cause a degradation in the catalyst. 
In particular, errors on the catalyst that would accumulate and increase over time would be those of the most challenging nature, as these directly impact the reusability of the catalyst.

\begin{enumerate}
    \item Errors on the initial catalyst $\tau$: 
    A simple computation shows that the catalyst does not degrade upon further iterations.
    Intuitively, this can be understood as follows: besides the initial error in catalyst preparation, there are no additional sources of errors during multiple rounds of catalysis, which use the same catalyst state.   
    In the meantime, the catalytic protocol stabilizes this catalyst state over time. 
    To be more precise, if the catalyst is initially prepared in a state $\tau_{\epsilon}$, which is $\epsilon$-away from the ideal catalyst state $\tau$, then the final state of the catalyst $\tau_{\epsilon}'$ remains at most $\epsilon$-away from $\tau$ via the data processing inequality, 
    \begin{align}
    \norm{\tau_{\epsilon}' - \tau}_1 = \norm{\Tr_S \bqty{\Lambda(\rho \otimes \tau_{\epsilon})} - \Tr_S \bqty{\Lambda(\rho \otimes \tau)}}_1  
    \leq \norm{\tau_{\epsilon} - \tau}_1 \leq \epsilon.
    \end{align}
    
	\item Errors on the initial state $\rho$ or the channel $\Lambda$:
	These errors are introduced afresh in each iteration, raising the possibility of error accumulation on the catalyst, despite the ideal protocol having a net-zero effect on the catalyst. We include a small explicit example below for illustration.
	Conceptually, if one possesses a complete characterization of the errors, e.g. knowing the exact perturbed input state $\rho_\epsilon$, it may still be possible to fine-tune the channel $\Lambda_\epsilon$ accordingly such that $\rho_\epsilon$ and $\Lambda_\epsilon$ stabilize the catalyst. 
	However, assuming full knowledge of such errors and the ability to adjust the channel accordingly is highly impracticable.
\end{enumerate}

\begin{example}[Accumulating errors on the catalyst] 
	It suffices to demonstrate that catalyst continues to degrade after the first round of catalysis. 
	Consider a qutrit system and a qubit catalyst undergoing a joint unitary evolution. 
	Denote the eigenvalues of the system and catalyst to be:
    \begin{equation}
        p_S = (p_1,p_2,p_3), \qquad q_C = (q_1,q_2).
    \end{equation}
    The unitary swaps the eigenstates corresponding to $p_1q_1 \leftrightarrow p_2q_2$ and $p_2q_1 \leftrightarrow p_3q_2$. 
    This operation is catalytic whenever $(p_1+p_2)q_1 = (p_2+p_3)q_2$. 
    Such toy examples are useful for illustrations, and have been used, e.g. in Appendix B of \cite{boes2020passing}. 

    Now, suppose that in the first round, we have a noisy system state $p_\varepsilon = (p_1-\varepsilon,p_2,p_3+\varepsilon)$ for some $\varepsilon >0$. 
    This leads to a final degraded catalyst $q' = (q_1',q_2')$, such that
    \begin{equation}
        q_1' = p_2q_2+(p_3+\varepsilon)q_2+ (p_3+\varepsilon)q_1 = q_1 + \varepsilon.
    \end{equation}
    From normalization, we also have that $q_2' = q_2 - \varepsilon$. In other words, the full amount of error $\varepsilon$ has propagated into the catalyst. 
    Next, suppose that in a second round, we have another noisy system state $p_{-\varepsilon} = (p_1+\varepsilon,p_2,p_3-\varepsilon)$. 
    Under the action of the same catalytic unitary, the catalyst further degrades into $\tilde q_C = (\tilde q_1,\tilde q_2)$, where
    \begin{align}
        \tilde q_1 &= p_2q_2' + (p_3-\varepsilon)q_2' + (p_3-\varepsilon)q_1' = q_1 - \varepsilon (p_2+p_3+q_2+q_1-p_3)\nonumber\\ 
        &= q_1 - \varepsilon(1+p_2).
    \end{align}
    In summary, the error accumulated almost linearly during two rounds of catalysis, as we anticipated. 
\end{example}

\section{For minimal, maximal, and separable compositions, resource non-generating operations (RNG) is completely resource non-generating operations(CRNG)}

\begin{lemma}\label{lem:RNGCRNG}
	If $\FSSp{AB} $ is either $\mincomp{A}{B}$, $\sepcomp{A}{B}$, or $\maxcomp{A}{B}$, then RNG = CRNG.
\end{lemma}
\begin{proof}
	We first consider the case $\FSSp{AB} = \mincomp{A}{B}$.
	For any system $R$, a free state $\gamma_{AR}\in\FSSp{AR}$ can be written as $\gamma_{AR} = \sum_{i}p_{i}(\gamma^{(i)}_{A}\otimes\gamma^{(i)}_{R})$, where $\gamma^{(i)}_{R}$ are free states for $R$.
	Let $\Lambda\in\CSp{A}{A'}$ be an RNG channel.
	Then the extension
	\begin{align}
		\Lambda\otimes\id_{R}(\gamma_{AR}) = \sum_{i}p_{i}(\tilde{\gamma}^{(i)}_{A'}\otimes\gamma^{(i)}_{R}),
	\end{align} 	
	where each $\tilde{\gamma}^{(i)}_{A'} = \Lambda(\gamma^{(i)}_{A}) \in\FSSp{A'}$.
	This implies that $\Lambda\otimes\id_{R}$ is an RNG channel and thus $\Lambda$ is a CRNG channel.
	
	Now we prove the case $\FSSp{AB} = \maxcomp{A}{B}$. 
	Again, let $\gamma_{AR}\in\FSSp{AR}$ be any free state. 
	The extension $ \tilde{\gamma}_{AR} = \Lambda\otimes\id_{R}(\gamma_{AR}) $ is also free if and only if its reduced states are free. 
	Since $\Lambda$ is an RNG channel, the $A$ reduced state $\Tr_{R}[\tilde{\gamma}] = \Lambda(\Tr_{R}[\gamma_{AR}])$ is free.
	The $R$ reduced state $\Tr_{A}[\tilde{\gamma}] = \Tr_{A}[\gamma_{AR}]$ is free because $\gamma_{AR}\in\FSSp{AR}$.
	Therefore, $\Lambda$ is a CRNG channel. 
	
	The proof is very similar for $\FSSp{AB} = \sepcomp{A}{B}$. 
	The free state $\gamma_{AR} = \sum_{i}p_{i}(\xi_{A}^{(i)}\otimes\zeta_{R}^{(i)})$ for $\xi_{A}^{(i)}\in\SSp{A}$, $\zeta_{R}^{(i)}\in\SSp{R}$, $p_{i}\geq 0$ for all $i$ and 
	$\sum_{i}p_{i}\xi_{A}^{(i)}\in\FSSp{A}$, $\sum_{i}p_{i}\zeta_{R}^{(i)}\in\FSSp{R}$.
	The final state after the extended channel becomes
	\begin{align}
		\Lambda\otimes\id_{R}(\gamma_{AR}) = \sum_{i}p_{i}(\tilde{\xi}^{(i)}_{A'}\otimes\zeta^{(i)}_{R}),
	\end{align} 
	where $\tilde{\xi}^{(i)}_{A'}\in\SSp{A'}$ for all $i$. 
	Furthermore, since $\Lambda$ is an RNG channel, we have that $\sum_{i}p_{i}\tilde{\xi}_{A'}^{(i)}\in\FSSp{A'}$.
	Hence, $\Lambda\otimes\id_{R}(\gamma_{AR}) $ is separable and its reduced states are free, making it a free state in $\sepcomp{A'}{R}$. 
\end{proof}

\section{Piani's theorem and the proof of Theorem~\ref{thm:mincomp_nocat}}\label{app:proof_thm_mincomp_nocat}

The proof of Thm.~\ref{thm:mincomp_nocat} relies mainly on constructing a faithful and strongly super-additive monotone for theories with minimal composition. 
To proceed, we introduce two specific monotones that are critical to the analysis. 
First, for any convex CRNG resource theory, the relative entropy of resource, is defined as
\begin{align}
    \REM(\rho_{X}) \coloneq \inf_{\sigma_{X}\in\FSSp{X}}S(\rho_{X}\Vert\sigma_{X}), \qquad S(\rho\|\sigma)\coloneq\Tr(\rho\log\rho-\rho\log\sigma),
\end{align}
is a monotone, i.e. $\REM(\Lambda(\rho_{X}))\leq \REM(\rho_{X})$ for any $\rho_{X}$ and $\Lambda\in\FCSp{X}{Y}$. 
It is a faithful measure yielding $\REM(\rho_{X}) \geq 0$, with equality if and only if $\rho_{X}\in\FSSp{X}$.
	
Similarly, the relative entropy of resource under restricted measurements can be defined. 
Consider a quantum measurement $M$, described by positive operator-valued measures (POVMs) $\{E_{i}\}_{i}$ that are positive and sum to identity. 
Let $M(\rho)$ denote the probability vector whose components are the probabilities $\Tr[\rho E_{i}]$ of obtaining an outcome $i$.
The Kullback–Leibler divergence $D_{\rm KL}(M(\rho) \Vert M(\sigma))$ between two outcome probabilities vanishes if and only if $M(\rho) = M(\sigma)$.
To consider multiple measurements, let $\mbM$ be a set of quantum measurements of interest.
The relative entropy of resource under $\mbM$ is then defined as
 \begin{align}
\MREM(\rho_{X}) = \inf_{\sigma_{X}\in\FSSp{X}}\sup_{M\in\mbM}D_{\rm KL}(M(\rho_{X}) \Vert M(\sigma_{X})).
 \end{align}
The following theorem establishes a relationship between the resource monotones of a composite state $\rho_{XY}$ and its marginals.
 
\begin{thm}[Ref.~\cite{Piani2009MRelEnt}, Thm.~1]\label{thm:Piani}
	Suppose that the free state set $\FSSp{X}$ is convex for any system $X$ and that $\rho_{XY}\in\FSSp{XY}$ implies $\Tr_{Y}[\rho_{XY}]\in\FSSp{X}$ for any subsystems $X$ and $Y$.
	Let $\mbM$ be a set of measurements on $X$, and assume that for all $M\in\mbM$ with POVMs $\{E_{i}\}_{i}$ and for all $\gamma_{XY}\in\FSSp{XY}$, the post-measurement $Y$ marginal state $\frac{\Tr_{X}[E_{i}\gamma_{XY}]}{\Tr[E_{i}\gamma_{XY}]}\in\FSSp{Y}$.
	Then, for any $\rho_{XY}\in\SSp{XY}$, 
	\begin{align}
		\REM(\rho_{XY}) \geq \MREM(\Tr_{Y}[\rho_{XY}]) + \REM(\Tr_{X}[\rho_{XY}]).
	\end{align}
\end{thm}
This theorem serves as a tool to prove Thm.~\ref{thm:mincomp_nocat}.

First, note that $\MREM$ possesses convenient properties: \emph{i)} data processing implies that $\MREM(\rho_{X})\leq\REM(\rho_{X})$ for any $\rho_{X}$ and any $\mbM$. Furthermore, it is known that \emph{ii)} if $\mbM$ includes informationally complete POVMs, $\MREM$ is faithful~\cite{Piani2009MRelEnt}, i.e. $\MREM(\rho_{X}) \geq 0$ with the equality if and only if $\rho_{X}\in\FSSp{X}$, and \emph{iii)} if $\mbM$ encompasses all possible POVMs, the monotonicity $\MREM(\Lambda(\rho_{X}))\leq\MREM(\rho_{X})$ for any $\rho_{X}$ and $\Lambda\in\FCSp{X}{Y}$ holds from data processing inequality.

Now consider a catalyst $\tau_{C}$ and a dilation $\Lambda\in\FCSp{SC}{S'C}$ inducing a catalytic channel $\tilde{\Lambda}$.
For any free system state $\gamma_{S}\in\FSSp{S}$, denote $\chi_{S'C} = \Lambda(\gamma_{S}\otimes\sigma_{C})$, where $\Tr_{C}[\chi_{S'C}] = \tilde{\Lambda}(\gamma_{S})$ and $\Tr_{S'}[\chi_{S'C}] = \tau_{C}$. 
By monotonicity of $\REM$, 
\begin{align}\label{eq:SC_dataprocessing}
	\REM(\tau_{C}) = \REM(\gamma_{S}\otimes\tau_{C}) \geq \REM(\chi_{S'C}),
\end{align}
where the first equality follows from both appending and discarding a free state $\gamma_{S}$ being a free operation.

To apply Thm.~\ref{thm:Piani}, set $\FSSp{S'C}=\mincomp{S'}{C}$ and $\mbM$ to be the set of \emph{all measurements}. 
 \begin{itemize}
     \item The first requirement, that $\FSSp{X}$ is convex for any system $X$ and that $\rho_{XY}\in\FSSp{XY}$ implies $\Tr_{Y}[\rho_{XY}]\in\FSSp{X}$ for any subsystem $X$ and $Y$ are already imposed as axioms for our framework.
     
     \item The second requirement can be shown using the structure of $\mincomp{S'}{C}$: note that any $\gamma_{S'C}\in\mincomp{S'}{C}$ can be written as 
     \begin{align}
     \gamma_{S'C} = \sum_{i}p_{i}(\gamma_{S'}^{(i)}\otimes\tilde{\gamma}_{C}^{(i)}).
     \end{align}
     Then for any POVM element $E_{S'}$ that acts on system $S'$, 
	\begin{align}\label{eq:post_measurement_free}
		\frac{\Tr_{S'}\left[E_{S'}\gamma_{S'C}\right]}{\Tr\left[E_{S'}\gamma_{S'C}\right]} = \frac{\sum_{i}p_{i}\Tr[E_{S'}\gamma_{S'}^{(i)}]\tilde{\gamma}_{C}^{(i)}}{\sum_{i}p_{i}\Tr[E_{S'}\gamma_{S'}^{(i)}]} \eqcolon \sum_{i}\tilde{p}_{i}\tilde{\gamma}_{C}^{(i)},
	\end{align} 
	where $\{\tilde{p}_{i}\}_{i}$ are valid convex coefficients. 
	By convexity, the resulting state $\frac{\Tr_{S'}\left[E_{S'}\gamma_{S'C}\right]}{\Tr\left[E_{S'}\gamma_{S'C}\right]} $ remains in $\FSSp{C}$.
 \end{itemize}
	
Using Thm.~\ref{thm:Piani}, 
\begin{align}
	\REM(\chi_{S'C}) \geq \MREM(\tilde{\Lambda}(\rho_{S})) + \REM(\tau_{C}).
\end{align}
Combined with Eq.~\eqref{eq:SC_dataprocessing}, it follows that $\REM(\tau_{C}) \geq \MREM(\tilde{\Lambda}(\rho_{S})) + \REM(\tau_{C})$, or equivalently,
\begin{align}
	0 \geq \MREM(\tilde{\Lambda}(\gamma_{S})),
\end{align}
whenever $\REM(\tau_{C})<\infty$. 
The latter is guaranteed by the existence of a full rank free state in $\FSSp{C}$ (criterion 4 of our basic assumptions stated in the main text).
By the faithfulness of $\MREM$, we conclude that $\tilde{\Lambda}(\gamma_{S})\in\FSSp{S'}$ for any $\gamma_{S}\in\FSSp{S}$.
In other words, catalytic channels for resource theories with minimal composition are always free operations.

Note that in Ref.~\cite{Takagi2022CorrCat} it has been shown that the super-additive monotone, if exists, also restricts marginal or correlated catalysis that are not robust. 
Our result then implies that theories with the minimal composition cannot be trivialized via (non-robust) marginal or correlated catalysis.

\section{Miscellaneous resource theories}

\subsection{Affine compositions}

Sometimes, free state sets have a stronger condition than being convex. 
There is a subclass of resource theories whose sets of free states are affine, such as the resource theory of athermality, asymmetry, coherence, and imaginarity~\cite{Gour2017Affine}, i.e. the set of free state $\FSSp{S}$ satisfies 
\begin{align}\label{eq: affine free staet set}
	\mathrm{aff}\FSSp{S}\cap \SSp{S} \coloneq \lset\sum_{i}p_{i}\rho_{S}^{(i)} \,\bigg\vert\, \forall i,\ \rho_{S}^{(i)}\in\FSSp{S},\ p_{i}\in\mathbb{R},\ \sum_{i}p_{i} = 1\rset\cap \SSp{S}  =\FSSp{S}.
\end{align}
However, if we impose the minimal composition, even when the free state sets are affine for each subsystem, the composite set might not be affine. 
To accommodate such theories, we consider affine composition of free states: the composite free state set is defined as
\begin{align}\label{def:aff_comp}
	\affcomp{A}{B} = \mathrm{aff}\lset \rho_{A}\otimes\rho_{B} \,\vert\, \rho_{A}\in\FSSp{A},\ \rho_{B}\in\FSSp{B} \rset \cap \SSp{AB},
\end{align}
for the system $AB$, given free state sets for $A$ and $B$.
However, it is important to remark that $\affcomp{A}{B}$ might not satisfy the four basic assumptions in the main text, when $\FSSp{A}$ or $\FSSp{B}$ is not affine. 
Suppose that $\gamma_{AB}\in\affcomp{A}{B}$, i.e. $\gamma_{AB} = \sum_{i}p_{i}\gamma_{A}^{(i)}\otimes\gamma_{B}^{(i)}$, where $\{p_{i}\}_{i}$ is a set of affine coefficients and $\gamma_{X}^{(i)}$ are some free states for system $X = A,B$.
In general, $\Tr_{B}[\gamma_{AB}] = \sum_{i}p_{i}\gamma_{A}^{(i)}\notin\FSSp{A}$ if $\FSSp{A}$ is not affine, breaking the second assumption in the main text. Nevertheless, when the resource theory is affine, the set $\affcomp{A}{B}$ is a valid free state set satisfying all four assumptions, and we establish the result analogous to Thm.~\ref{thm:mincomp_nocat} in the main text. 
\begin{thm}\label{thm:affcomp_nocat}
	If an affine CRNG resource theory has the affine composition, then it does not allow non-free catalytic channels and resource broadcasting.    
\end{thm}
\begin{proof}
	The proof is identical to that of Thm.~\ref{thm:mincomp_nocat}, except for replacing convex coefficients $\{p_{i}\}_{i}$ and $\{\tilde{p}_{i}\}_{i}$ by affine coefficients. 
	This replacement does not change the conclusion, since the affine combination of free states are assumed to be free.  
\end{proof}

\subsection{Intersection of multiple completely resource non-generating operations}\label{subsec:GPC}

In some resource theories of interest, the free operation set is not the set of CRNG operations, but given as the intersection of multiple CRNG operation sets for different resources. 
In such cases, if each CRNG operation does not allow broadcasting of a resource, the intersection of them also has no-broadcasting property. 
A prominent example is the Gibbs-preserving covariant operations~\cite{Cwiklinski2015GPC, Gour2018GPC}, which is an intersection of Gibbs-preserving operations (athermality) and covariant operations (asymmetry).
Therefore, Gibbs-preserving covariant operations cannot benefit from robust catalysis. 

\subsection{Resource theory of imaginarity}\label{subsec:imaginarity}
Resource theory of imaginarity is defined by the set of free state
\begin{align}
	\FSSp{X} = \lset \rho_{X} \,\vert\, \rho_{X}\in\SSp{X},\ \bra{i}_{X}\rho_{X}\ket{j}_{X}\in\mathbb{R},\ \forall i,j \rset,
\end{align}
where $\{\ket{i}_{X}\}_{i}$ is a fixed basis for system $X$ prescribed by some restrictions.
When $S$ and $C$ are qubit systems, CRNG operation $\FCSp{SC}{SC}$ includes CNOT gate, which maps $\ket{0i}_{SC}\mapsto\ket{ii}_{SC}$ for $i = 0,1$. 

Suppose that the maximally imaginary state $\ket{\hat{+}}_{C} = \frac{1}{\sqrt{2}}(\ket{0}_{C} + i\ket{1}_{C})$ is given as a catalyst.
Then the strict robust catalysis
\begin{align}
	\Lambda(\dm{0}_{S}\otimes\dm{\hat{+}}_{C}) = \dm{\hat{+}}_{S}\otimes\dm{\hat{+}}_{C}
\end{align}
is implementable using a combination of CNOT and H gates $\Lambda$~\cite{Takagi2017Imaginarity}. The broadcasting version of this channel, $\mB(\dm{\hat{+}}) = \dm{\hat{+}}^{\otimes 2}$ is a special case of resource broadcasting, which is dubbed catalytic replication in Ref.~\cite{Kuroiwa2020catreplication}.
To develop more intuition, we invoke Proposition~1 of Ref.~\cite{Wu2021Imaginarity} stating that any pure state $\ket{\psi}$ can be transformed to an effectively qubit pure state $\ket{\phi} = \alpha\ket{0} + i\beta\ket{1}$ with $\alpha,\beta\in\mathbb{R}$, via some real unitary operation. 
Since the operation is unitary, the inverse of such operation is also a real operation, i.e. any pure state $\ket{\psi}$ is equivalent to some qubit pure state $\ket{\phi}$ in terms of the imaginarity. 
It also means that any pure state cannot have a resource exceeding that of the maximally imaginary pure qubit state $\ket{\hat{+}}$, even if the former state consists of multiple copies of the latter state. 

On the other hand, this non-extensiveness does not indicate that the catalytic replication $\rho\to\rho^{\otimes 2}$ is always possible via real operations. 
In Ref.~\cite{Zhang_2024a}, it is shown that the maximally imaginary state is the only state that admits catalytic replication among qubit states or pure states. 

\subsection{Limited subspace theories}\label{subsection:non-fullrank}
It is easy to construct a theory that admits robust catalysis if we break the fourth assumption in the main text, that is, if all free states are non-full rank.
Suppose that there exists a catalyst state $\tau_{C}$ that is not in the support of any free state in $\FSSp{C}$. 
Then a catalytic channel can be constructed by the following steps.
First, the catalyst part is measured to distinguish whether it is in the support of some free state.
If it is in the support, prepare a system-catalyst free state. 
If it is not, which is the case for $\tau_{C}$ we assumed, prepare $\sigma_{S}\otimes\tau_{C}$, where $\sigma_{S}$ is any system state and $\tau_{C}$ is the catalyst reduced state of the initial state. 

The simplest example is the resource theory of athermality at temperature $T=0$, where the only free state is the ground state $\gamma_{X} = \dm{0}_{X}$ for any system $X$.
If a catalyst $\tau_{C} = \dm{1}_{C}$, the channel 
\begin{align}
	\Lambda(\rho_{SC}) = \dm{00}_{SC}\Tr[\rho_{SC}(\1_{S}\otimes\gamma_{C})] + \sigma_{S}\otimes \dm{1}_C \Tr[\rho_{SC}(\1_{S}\otimes[\1_{C} - \gamma_{C}])] 
\end{align}
transforms $\Lambda(\dm{00}_{SC}) = \dm{00}_{SC}$ and $\Lambda(\rho_{S}\otimes\tau_{C}) = \sigma_{S}\otimes\tau_{C}$, regardless of the choices for $\rho_{S},\sigma_{S}\in\SSp{S}$.
Hence, any state transformation becomes possible with a (strict) robust catalysis. 
A very similar construction was used in Ref.~\cite{Kuroiwa2020catreplication}.

\section{Generalizations and limitations of Theorem~\ref{thm:robcat_possible}}\label{section: robcat possible generalization}

\subsection{Proof of Theorem~\ref{thm:robcat_possible}}

We prove Thm.~\ref{thm:robcat_possible} by proving a slightly more general theorem. 
Before that, we first define another class of composite free state set
\begin{align}\label{def:sep_comp}
	\sepcomp{A}{B}\coloneq (\maxcomp{A}{B}) \cap \mathrm{SEP},
\end{align}
where $\mathrm{SEP}$ represents all separable states across the $A \vert B$ partition.
This set explicitly excludes entanglement between subsystems as a free resource.
We prove that catalytic channels are advantageous even with Eq.~\eqref{def:sep_comp}, thereby demonstrating that the robust catalytic advantage is not merely an artifact of free entanglement.

\begin{thm}\label{thm:robcat_possible_general1}
	Suppose the free state set of system $C$ is a singleton $\FSSp{C} = \{\gamma_{C}\}$ and the composite free state set $\FSSp{SC}$ is either $\maxcomp{S}{C}$ or $\sepcomp{S}{C}$.
	Then there exists a free operation $\Lambda\in\FCSp{SC}{SC}$ and a catalyst $\tau_{C}\in\SSp{C}$ such that 
	\begin{align}
		\tau_{C} &= \Tr_{S}[\Lambda(\rho_{S}\otimes\tau_{C})],\\
		\sigma_{S} &= \Tr_{C}[\Lambda(\rho_{S}\otimes\tau_{C})],
	\end{align}
	for all $\rho_{S}\in\SSp{S}$, if and only if 
	\begin{align}\label{eq:dma_cond_suppl} 
		\sup_{\tau_{C}\in\SSp{C}} D_{\max} (\tau_{C}\Vert\FSSp{C}) \geq D_{\max}(\sigma_{S}\Vert\FSSp{S}).
	\end{align}
	The existence of such $\Lambda$ implies the existence of a catalytic measure-and-prepare channel $\tilde{\Lambda}(\bullet_{S}) = \Tr_{C}[\Lambda(\bullet_{S}\otimes\tau_{C})]$	preparing the state $\sigma_{S}$ and a broadcasting channel $\mB(\bullet_{C}) = \Lambda (\gamma_{S}\otimes\bullet_{C}) \in\FCSp{C}{SC}$ with some $\gamma_{S}\in\FSSp{S}$.
\end{thm}

\begin{proof}
We first show the necessity of Eq.~\eqref{eq:dma_cond_suppl}.
Suppose that a broadcasting channel $\mB(\bullet_{C}) = \Lambda (\gamma_{S}\otimes\bullet_{C}) $ yields $\Tr_{S}[\mB(\tau_{C})] = \tau_{C}$ and $\Tr_{C}[\mB(\tau_{C})] = \sigma_{S}$.
Since $D_{\max}$ is a resource monotone, it never increases after a free operation, i.e.
\begin{align}
	D_{\max}(\tau_{C}\Vert\FSSp{C})\geq D_{\max}(\mB(\tau_{C})\Vert\FSSp{SC})\geq D_{\max}(\sigma_{S}\Vert\FSSp{S}).
\end{align}
Hence, for any $\tau_{C}$, Eq.~\eqref{eq:dma_cond_suppl} is satisfied.

To prove sufficiency, note that $\max_{\tau_{C}\in\SSp{C}}D_{\max} (\tau_{C}\Vert\FSSp{C}) = -\log(\lambda_{1})$ is attained for $\psi_{C} = \dm{\psi}_{C}$ corresponding to the smallest eigenvalue of the catalyst free state $\lambda_{1}$.
Consider the measure-and-prepare channel
\begin{align}
	\Lambda(\varrho_{SC}) \coloneq \Tr[(\1_{S}\otimes\psi_{C})\varrho_{SC}] (\sigma_{S}\otimes\psi_{C})+ \Tr[(\1_{SC}-\1_{S}\otimes\psi_{C})\varrho_{SC}] (\omega_{S}\otimes\zeta_{C}),
\end{align}
which maps $\rho_{S}\otimes\psi_{C}\to\sigma_{S}\otimes\psi_{C}$.
To complete the proof, it remains to verify $\Lambda\in\FCSp{SC}{SC}$.
Using Lemma~\ref{lem:RNGCRNG} in~\suppl, it is sufficient to show that $\Lambda$ is RNG.
Any free state $\gamma_{SC}\in\FSSp{SC}$ has the marginal state $\gamma_{C}$, which implies
\begin{align}
	\Lambda(\gamma_{SC}) = \lambda_{1}(\sigma_{S}\otimes\psi_{C}) + (1-\lambda_{1})(\omega_{S}\otimes\zeta_{C}).
\end{align}
The $C$ reduced state $\Tr_{S}[\Lambda(\gamma_{SC})]$ can always be made free by choosing $(1-\lambda_{1}) \zeta_{C} = \gamma_{C} - \lambda_{1}\psi_{C}$, while the $S$ reduced state $\Tr_{C}[\Lambda(\gamma_{SC})]$ is free if $\lambda_{1}\sigma_{S}+(1-\lambda_{1})\omega_{S}\in\FSSp{S}$.
The latter is equivalent to the fact that there exists $\gamma_{S}\in\FSSp{S}$, such that $\gamma_{S}-\lambda_{1}\sigma_{S}\geq0$, i.e. $D_{\max}(\sigma_{S}\Vert\FSSp{S})\leq-\log(\lambda_{1}) = D_{\max}(\psi_{C}\Vert\FSSp{c})$.
If that is the case, $\Lambda\in\FCSp{SC}{SC}$ when $\FCSp{SC}{SC}$ is defined by the maximal composition $\FCSp{SC}{SC} = \textrm{CRNG}(\maxcomp{S}{C})$.
Furthermore, $\Lambda(\gamma_{C})$ is a separable operation.
Hence, when $\FCSp{SC}{SC}$ is defined by the separable composition $\FCSp{SC}{SC} = \textrm{CRNG}(\sepcomp{S}{C})$, $\Lambda$ is also free. 
\end{proof}

\subsection{Resource broadcasting for commuting sets}

We establish a theorem analogous to Theorem~\ref{thm:robcat_possible_general1} for another class of free states $\FSSp{C}$.
In particular, the necessary and sufficient condition Eq.~\eqref{eq:dmax_cond_commuting} is identical to that of Theorem~\ref{thm:robcat_possible_general1}.

\begin{thm}\label{thm:robcat_possible_general2}
	Suppose that the maximally mixed state $\frac{\1_{C}}{d_{C}}\in\FSSp{C}$, where $d_{C} = \dim(\1_{C})$, and all states in $\FSSp{C}$ commute with each other.
	Additionally assume that $\FSSp{SC} = \maxcomp{S}{C}$ or $\sepcomp{S}{C}$.  
	Then there exists a state $\tau_{C}\in\SSp{C}$ and a broadcasting channel $\mB\in\FCSp{C}{SC}$, such that $\sigma_{S} = \Tr_{C}[\mB(\tau_{C})]$ can be prepared in $S$, if and only if
	\begin{align}\label{eq:dmax_cond_commuting}    
		\sup_{\tau_{C}\in\SSp{C}} D_{\max} (\tau_{C}\Vert\FSSp{C}) \geq D_{\max}(\sigma_{S}\Vert\FSSp{S}).
	\end{align}
\end{thm}

\begin{proof}
        We can show that the condition is necessary by following the proof in Theorem~\ref{thm:robcat_possible}.
        Now, let us show the condition is sufficient.
	Since $D_{\max} (\tau_C \Vert \FSSp{C})$ is a convex function of $\tau_C$, we note that $\sup_{\tau_{C}\in\SSp{C}} D_{\max} (\tau_{C}\Vert\FSSp{C}) = D_{\max} (\psi_{C}\Vert\FSSp{C})$, for some pure state $\psi_{C} = \dm{\psi}_{C}$.
        Furthermore, since all $\FSSp{C}$ states are commuting by assumption, there exists an orthonormal basis $\{\ket{i}_{C}\}_{i}$, in which all states in $\FSSp{C}$ are diagonal.
	We will show that the optimal pure state is $\ket{\psi}_{C} = \sum_{i}\frac{1}{\sqrt{d_{C}}}\ket{i}_{C}$.
	Any state $\gamma_{C}\in\FSSp{C}$ can be written as $\gamma_{C} = \sum_{i}p_{i}\dm{i}_{C}$ and thus 
	\begin{align}\label{eq: overlap with pure catalyst}
		\bra{\psi}_{C}\gamma_{C}\ket{\psi}_{C} = \frac{1}{d_{C}}.
	\end{align}
	This means that $r\gamma_{C} - \dm{\psi}_{C} \geq 0$ only when $r\geq d_{C}$, i.e. $D_{\max} (\psi_{C}\Vert\FSSp{C}) = \log d_{C}$.
	On the other hand, $d_{C}\frac{\1_{C}}{d_{C}} - \tau_{C} \geq 0$ for any state $\tau_{C}\in\SSp{C}$. 
	Which indicates that $D_{\max} (\tau_{C}\Vert\FSSp{C}) \leq \log d_{C}$ for any $\tau_{C}\in\SSp{C}$. 
	Therefore, the maximum for the LHS in Eq.~\eqref{eq:dmax_cond_commuting} is obtained when $\tau_{C} = \dm{\psi}_{C}$.
	
	Now we show that $\dm{\psi}_{C}$ can broadcast any $\sigma_{S}$, such that $\log d_{C} \geq D_{\max}(\sigma_{S}\Vert\FSSp{S})$.
	Consider the same map used in Appendix E in End Matters:
	\begin{align}\label{eq: broadcasting channel for general thm 4}
		\mB(\rho_{C}) \coloneq \Tr[\psi_{C}\rho_{C}] (\sigma_{S}\otimes\psi_{C}) + \Tr[(\1_{C}-\psi_{C})\rho_{C}] (\omega_{S}\otimes\zeta_{C}),
	\end{align}
	which broadcasts $\psi_{C}\to\sigma_{S}\otimes\psi_{C}$.
	Note that this is also a strict broadcasting.
	We would like to show that $\mB$ is RNG. 
	
	By definition of $D_{\max}$, there exists a free state $\tilde{\gamma}_{S}\in\FSSp{S}$ and a state $\omega_{S}\in\SSp{S}$, such that $(d_{C}\tilde{\gamma}_{S} - \sigma_{S}) \propto \omega_{S}$.
	From the normalization of quantum states, we have
	\begin{align}\label{eq: Dmax alternate form for S}
		\tilde{\gamma}_{S} = \frac{1}{d_{C}}\sigma_{S} + (1 - \frac{1}{d_{C}})\omega_{S}.
	\end{align}
	Because of Eq.~\eqref{eq: overlap with pure catalyst}, we have $\Tr[\psi_{C}\gamma_{C}] = \frac{1}{d_{C}}$ and $\Tr[(\1_{C}-\psi_{C})\gamma_{C}] = 1 - \frac{1}{d_{C}}$ for any $\gamma_{C}\in\FSSp{C}$.
	By setting $\zeta_{C} = \frac{\1_{C} - \psi_{C}}{d_{C}-1}$ and $\omega_{S}$ from Eq.~\eqref{eq: Dmax alternate form for S}, we have
	\begin{align}
		\mB(\gamma_{C}) =  \frac{1}{d_{C}}(\sigma_{S}\otimes\psi_{C}) +  (1 - \frac{1}{d_{C}}) (\omega_{S}\otimes\zeta_{C}),
	\end{align}
	such that $\Tr_{C}[\mB(\gamma_{C})] = \tilde{\gamma}_{S}\in\FSSp{S}$ and $\Tr_{S}[\mB(\gamma_{C})] =\frac{\1_{C}}{d_{C}}\in\FSSp{C}$, which proves that $\mB(\gamma_{C})\in\maxcomp{S}{C}$ and $\mB(\gamma_{C})\in\sepcomp{S}{C}$.
\end{proof}

One extreme example is the resource theory of local coherences, where all states that are diagonal in a basis $\{\ket{i}_{C}\}_{i}$ are considered free.
For qubit subsystems, the local free state sets are given as incoherent states $\FSSp{S(C)}=\{p\dyad{0}+(1-p)\dyad{1}|0\leq p \leq 1\}$.
We choose the composite free state set $\FSSp{SC}$ to be $\maxcomp{S}{C}$.
This setting of local coherence has been studied in~\cite{Horova_2022}.

We illustrate how Thm.~\ref{thm:robcat_possible_general2} works in this setup. 
This means that any $\sigma_{S}\in\SSp{S}$ is attainable, since $\FSSp{S}$ and $\FSSp{C}$ are the same.
First note that $\sup_{\tau_{C}\in\SSp{C}} D_{\max} (\tau_{C}\Vert\FSSp{C}) = D_{\max} (\dm{+}_{C}\Vert \FSSp{C}) = \log2$, where $\ket{+} = \frac{1}{\sqrt{2}}(\ket{0}+\ket{1})$.

The channel $\mB$ corresponding to Eq.~\eqref{eq: broadcasting channel for general thm 4} is
\begin{align}
	\mB(\rho_{C}) = \bra{+}_{C}\rho_{C}\ket{+}_{C} (\sigma_{S}\otimes\dm{+}_{C}) + \bra{-}_{C}\rho_{C}\ket{-}_{C} (\omega_{S}\otimes\dm{-}_{C}).
\end{align}
For this case, catalytic replication~\cite{Kuroiwa2020catreplication} is also possible, i.e. $\mB(\dm{+}_{C}) = \dm{+}_{C}\otimes\dm{+}_{C}$ with $\mB\in\FCSp{C}{SC}$.

\subsection{Local resource theories}
We can extend Theorem~\ref{thm:robcat_possible_general2} to other settings that is concerned with local resources such as local magic and local entanglement.
As an example, let us focus on local magic.
Here, the local system is a qubit, with the local free state set given by the convex hull of stabilizer states $\mathcal{S}_X$, and the composite free states are either $\mathcal{S}_{SC} = \maxcomp{S}{C}$ or $\sepcomp{S}{C}$.
Then, any state preparation channel can be implemented as a catalytic channel.
\begin{thm}
  For any state $\sigma_S$, there exists a state $\tau_{C}\in\SSp{C}$ and a broadcasting channel $\mB\in\FCSp{C}{SC}$, such that $\sigma_{S} = \Tr_{C}[\mB(\tau_{C})]$.
\end{thm}
\begin{proof}
  Let us pick a fixed but arbitrary $\sigma$.
  Let $\ket{T} = 1/\sqrt{2} \pqty{\ket{0} + \exp\pqty{i \pi/4} \ket{1}}$, and let $\ket{\tilde{T}} = 1/\sqrt{2} \pqty{\ket{0} - \exp\pqty{i \pi/4} \ket{1}}, \tilde{\sigma}$ be the state opposite to $\ket{T}, \sigma$ on the Bloch sphere.
  Let
  \begin{align}
      \mB(\rho_{C}) &\coloneq \bra{T}_{C}\rho_{C}\ket{T}_C (\sigma_{S}\otimes\ketbra{T}_{C}) + \bra{\tilde{T}}_{C}\rho_{C}\ket{\tilde{T}}_C (\tilde{\sigma}_{S}\otimes\ketbra{\tilde{T}}_{C}).
  \end{align}
  Observe that if $\sin^2 \pqty{\pi/8} \leq p \leq \cos^2 \pqty{\pi/8}$, then $(1-p) \rho + p \tilde{\rho}$ is always a stabilizer state for any state $\rho$.
  Now, if $\sigma$ is a stabilizer state, then we have $\sin^2 \pqty{\pi/8} \leq \bra{T} \sigma \ket{T} \leq \cos^2 \pqty{\pi/8}$ and therefore $\mathcal{B} (\sigma)$ is in $\mathcal{S}_{SC}$.
  Furthermore, $\mathcal{B} (\ketbra{T}_C) = \sigma_S \otimes \ketbra{T}_C$, proving the claim.
\end{proof}

\subsection{No broadcasting under any composition}

In this section, we show that Theorem~\ref{thm:robcat_possible} cannot be generalized to all theories, even under maximal/separable compositions. We formulate this in terms of a no-broadcasting theorem.
\begin{thm}
Consider free state sets $\FSSp{S},\FSSp{C}$ such that:
\begin{enumerate}[leftmargin=*,itemsep=0pt]
    \item The affine hull of $\FSSp{C}$ contains $\SSp{C}$, in other words, any $\tau_{C}\in\SSp{C}$ can be written as an affine combination of elements in $\FSSp{C}$.
    \item $\FSSp{S}$ is affine.
\end{enumerate}
Then broadcasting is impossible under any composition rule, as long as the resource theory satisfies the basic assumption (A2) in main text.
\end{thm}

Before starting the proof, let us give examples of the above two conditions in the theorem. An example of $\FSSp{C}$ satisfying condition 1 would be the set of separable states. Meanwhile, an example of $\FSSp{S}$ satisfying condition 2 would be a singleton set, $ \FSSp{S} = \{\frac{\1_{S}}{d_{S}}\}$.

The proof strategy is to show that for \emph{any} state $\tau_C$ and any free channel $\mB \in \FCSp{C}{SC}$, the state $\Tr_{C} [\mathcal{B}(\tau_C)]$ is always free.
This means that there are states $\sigma_S$ that satisfy
\begin{align}
    \sup_{\tau_{C}\in\SSp{C}} D_{\max} (\tau_{C}\Vert\FSSp{C}) \geq D_{\max}(\sigma_{S}\Vert\FSSp{S}),
\end{align}
and yet $\sigma_S \neq \Tr_C \mathcal{B}(\tau_C)$ for any $\tau_C$ and any free $\mathcal{B}$.

\begin{proof}
Firstly, note that for any free map $\mB \in \FCSp{C}{SC}$ and any free state $\gamma_{C}$, assumption (A2) in the main text implies that we must have $\Tr_{C} [\mathcal{B}(\gamma_C)] \in \FSSp{S}$.
However, by condition 1, any density matrix on $C$ can be written as an affine combination $\tau_{C} = \sum_{i} \alpha_{i} \gamma_{i}$ with real coefficients $\alpha_{i}$ and free states $\gamma_{i}\in\FSSp{C}$.
Next, linearity of $\mB$ gives $\Tr_{C} [\mB (\tau_{C})] = \sum_{i} \alpha_{i} \Tr_{C} \mB (\gamma_i)$.
Since $\mathcal{S}_S$ is an affine set by condition 2, $\Tr_{C} [\mB (\tau_{C})] \in \FSSp{S}$.
\end{proof}

Note that in the above reasoning, the composition rule is left unspecified;
as long as $\Tr_C\circ\mB(\FSSp{C}) \subseteq\FSSp{S}$, the claim holds.
We end this section with a final complementary remark: while $C \to SC$ broadcasting is not allowed in this theorem, $S \to SC$ broadcasting may still be possible, e.g. if the condition in Theorem~\ref{thm:robcat_possible} or Theorem~\ref{thm:robcat_possible_general2} is fulfilled.

\newcommand{\revrelent}[1]{\scalebox{-1}[1]{R}\pqty{#1}}
\section{Strict robust catalysis}\label{sec:strictRC}

We show that strict robust catalysis cannot be advantageous when using a full rank catalyst state. 
We begin by defining the reversed relative entropy of resource, first used in Ref.~\cite{Eisert_2003} in the context of entanglement theory,
\begin{align}
	\revrelent{\rho_{S}} \coloneq \inf_{\gamma_{S} \in \FSSp{S}} D(\gamma_{S} \Vert \rho_{S}).
\end{align}
It is easy to verify that this is a faithful measure, i.e. $\revrelent{\rho_{S}}\geq0$ with equality if and only if $\rho_{S}\in\FSSp{S}$ and that $\revrelent{\rho_{S}}<\infty$ when $\rho_{S}$ is full rank.

\begin{lemma}
  The reversed relative entropy of resource is additive, i.e.
  \begin{align}
	\revrelent{\rho_{S}\otimes\omega_{S'}} = \revrelent{\rho_{S}} + \revrelent{\omega_{S'}}.
  \end{align}
\end{lemma}
\begin{proof}
  Recall the following property of the quantum relative entropy: $D(\gamma_{SS'} \Vert \rho_{S}\otimes\omega_{S'}) \geq D(\Tr_{S'}[\gamma_{SS'}] \Vert \rho_{S}) + D(\Tr_{S}[\gamma_{SS'}] \Vert \omega_{S'}) $, with the equality if and only if $\gamma_{SS'}$ is an uncorrelated state. 
  Since for any $\gamma_{SS'} \in \FSSp{SS'}$, the uncorrelated state $\Tr_{S'}[\gamma_{SS'}]\otimes\Tr_{S}[\gamma_{SS'}]\in\FSSp{SS'}$, the infimum of $ D(\gamma_{SS'} \Vert \rho_{S}\otimes\omega_{S'})$ over $\gamma_{SS'} \in \FSSp{SS'}$ is always obtained when $\gamma_{SS'}$ is uncorrelated, and we obtain the claim.
\end{proof}

Now, suppose a free state $\gamma_{S}\in\FSSp{S}$ can be transformed into another state $\sigma_{S'}$ via strict robust catalysis with a full rank catalyst $\tau_{C}$. 
Then there exists a channel $\Lambda\in\FCSp{SC}{S'C}$, such that 
\begin{align}
	\Lambda(\gamma_{S}\otimes\tau_{C}) = \sigma_{S'}\otimes\tau_{C}.
\end{align}
From the additivity and the monotonicity, 
\begin{align}
	\revrelent{\tau_{C}} = \revrelent{\gamma_{S}\otimes\tau_{C}} \geq \revrelent{ \sigma_{S'}\otimes\tau_{C}} = \revrelent{\sigma_{S'}} + \revrelent{\tau_{C}}.
\end{align}
By assumption, $\tau_{C}$ is full-rank, so $\revrelent{\tau_{C}}$ is finite, which implies that $0\geq \revrelent{\sigma_{S'}}$, or equivalently $\sigma_{S'}\in\FSSp{S'}$ from the faithfulness.

\section{Robust battery-assisted transformation}

Let us investigate the robustness of battery-assisted transformations, introduced in~\cite{arxiv_Ganardi_2024,Alhambra_2019}.
In this framework, there is a distinguished measure of resource $R$.
A transformation from $\rho_{S}$ to $\sigma_{S'}$ is said to be possible with a battery if there exists a battery state $\tau_{C}$ and a channel $\Lambda\in\FCSp{SC}{S'C'}$ such that $\Lambda(\rho_{S} \otimes \tau_{C}) = \sigma_{S'} \otimes \tau'_{C'}$ with $R(\tau'_{C'}) \geq R(\tau_{C})$.
Let us define a robust analogue as follows:
a $(\rho,\epsilon)$-robust transformation is possible if $\Lambda(\rho_{S} \otimes \tau_{C}) = \sigma_{S'} \otimes \tau'_{C'}$ with $R(\tau'_{C'}) \geq R(\tau_{C})$, and for all states $\rho'_{S}$, such that $\Vert \rho'_{S} - \rho_{S} \rVert_{1}\leq \epsilon$, we have 
\begin{align}
	R\left(\Tr_{S}\left[\Lambda(\rho'_{S}\otimes\tau_{C})\right]\right) \geq R(\tau_{C})
\end{align}
In addition, we can allow arbitrarily small error in the target state as long as the resource in the battery is uniformly bounded.

Ref.~\cite{arxiv_Ganardi_2024} showed that if $R$ is a finite and additive monotone, then $\rho_{S}$ can be transformed into $\sigma_{S'}$ with a battery if and only if $R(\rho_{S}) \geq R(\sigma_{S'})$.
We now demonstrate that such transformations are robust when $R$ is convex and continuous.
If $R(\rho_{S}) > R(\sigma_{S'})$, then by continuity, there exists an $\epsilon > 0$ such that for all $\rho'_{S}$ satisfying $\Vert \rho'_{S} - \rho_{S} \rVert_{1}\leq \epsilon$, we have $R(\rho'_{S}) \geq R(\sigma_{S'})$.
Following the proof in Ref.~\cite{arxiv_Ganardi_2024}, we can show that the transformation is robust.
If $R(\rho_{S}) = R(\sigma_{S'})$, convexity ensures that for any $\delta > 0$, there exists a state $\sigma^{\delta}_{S'}$, such that $\Vert \sigma^{\delta}_{S'} - \sigma_{S'} \rVert_{1}\leq \delta$ and $R(\sigma^{\delta}_{S'}) \leq (1-\delta) R(\sigma_{S'}) = (1-\delta) R(\rho_{S}) < R(\rho_{S})$.
Consequently, for any sequence $\Bqty{\delta_n} \to 0$, there exists a robust battery-assisted transformation from $\rho_{S}$ to $\sigma^{\delta_n}_{S'}$, and $\sigma^{\delta_n}_{S'} \to \sigma_{S'}$.
Furthermore, the resource stored in the battery can be assumed to be uniformly bounded by $R(\sigma_{S'})$.
Therefore, we conclude that there exists a robust battery-assisted transformation from $\rho_{S}$ to $\sigma_{S'}$.

\end{document}